\documentclass[a4paper,english,runningheads,cleveref]{lipics-v2021} 
\usepackage{xspace}
\usepackage{amssymb,amsfonts,amsmath,amsthm}
\usepackage{todonotes}
\usepackage{enumerate}
\usepackage{soul}
\usepackage{hyperref}
\usepackage[capitalise]{cleveref}
\usepackage{paralist}
\usepackage{lineno}[pagewise]
\usepackage{graphicx}
\usepackage[inline]{enumitem}
\usepackage[bibliography=common]{apxproof}
\ifx\proof\inlineproof
  \let\apxmark\relax
  
\else
  \def\apxmark{\textcolor{lipicsGray}{\sf\bfseries$\!$(*)\,}}
  
\fi

\nolinenumbers

\newtheoremrep{theorem2}[theorem]{Theorem}
\newtheoremrep{corollary2}[theorem]{Corollary}
\newtheoremrep{lemma2}[theorem]{Lemma}
\newtheoremrep{claim2}[theorem]{Claim}
\newtheoremrep{property2}[theorem]{Property}
\newenvironment{appproof}{\textbf{\textcolor{lipicsGray}{Proof.}}}{\hfill$\square$\medskip}
\renewcommand{\appendixsectionformat}[2]{Omitted Material from Section~#1}

\Crefformat{figure}{#2Fig.~#1#3}
\author{Patrizio~Angelini}{John Cabot University, Rome, Italy}{pangelini@johncabot.edu}{}{}

\author{Michael~A.~Bekos}{Department of Mathematics, University of Ioannina, Ioannina, Greece}{bekos@uoi.gr}{}{}

\author{Julia~Katheder}{Wilhelm-Schickard-Institut f{\"u}r Informatik, Universit{\"a}t T{\"u}bingen, T{\"u}bingen, Germany}{julia.katheder@uni-tuebingen.de}{}{}

\author{Michael~Kaufmann}{Wilhelm-Schickard-Institut f{\"u}r Informatik, Universit{\"a}t T{\"u}bingen, T{\"u}bingen, Germany}{mk@informatik.uni-tuebingen.de}{}{}

\author{Maximilian~Pfister}{Wilhelm-Schickard-Institut f{\"u}r Informatik, Universit{\"a}t T{\"u}bingen, T{\"u}bingen, Germany}{maximilian.pfister@uni-tuebingen.de}{}{}

\authorrunning{Angelini, Bekos, Katheder, Kaufmann, Pfister}

\Copyright{Patrizio~Angelini, Michael A. Bekos, Julia~Katheder, Michael~Kaufmann, Maximilian~Pfister}


\EventEditors{Stefan Szeider, Robert Ganian, Alexandra Silva}
\EventNoEds{3}
\EventLongTitle{
	47th International Symposium on
	Mathematical Foundations of Computer Science}
\EventShortTitle{MFCS 2022}
\EventAcronym{MFCS}
\EventYear{2022}

\EventLogo{}

\title{RAC Drawings of Graphs with Low Degree}

\newtheorem{prop}{Property}

\newtheorem{invariant}{Invariant}
\newtheorem{question}{Question}
\graphicspath{{images/}}
\begin{document}
\maketitle

\newcommand{\verygood}{kite-planar $1$-planar\xspace}
\newcommand{\good}{$1$-kite-planar\xspace}
\newcommand{\piece}{piece of a kite\xspace}

\begin{abstract}
Motivated by cognitive experiments providing evidence that large crossing-angles do not impair the readability of a graph drawing, RAC (Right Angle Crossing) drawings were introduced to address the problem of producing readable representations of non-planar graphs by supporting the optimal case in which all crossings form $90^\circ$ angles.

In this work, we make progress on the problem of finding RAC drawings of graphs of low degree. In this context, a long-standing open question asks whether all degree-$3$ graphs admit straight-line RAC drawings. This question has been positively answered for the Hamiltonian degree-3 graphs. We improve on this result by extending to the class of $3$-edge-colorable degree-$3$ graphs. When each edge is allowed to have one bend, we prove that degree-$4$ graphs admit such RAC drawings, a result which was previously known only for degree-$3$ graphs. Finally, we show that $7$-edge-colorable degree-$7$ graphs admit RAC drawings with two bends per edge. This improves over the previous result on degree-$6$ graphs. 
\end{abstract}

\section{Introduction}

In the literature, there is a wealth of approaches to draw planar graphs. Early results date back to Fáry’s theorem~\cite{Far48}, which guarantees the existence of a planar straight-line drawing for every planar graph; see also~\cite{CON85,St51,SH34,Tu63,Wag36}. Over the years, several breakthrough results have been proposed, e.g., de Fraysseix, Pach and Pollack~\cite{DBLP:conf/stoc/FraysseixPP88} in the late 80’s devised a linear-time algorithm~\cite{DBLP:journals/ipl/ChrobakP95} that additionally guarantees the obtained drawings to be on an integer grid of quadratic size (thus making high-precision arithmetics of previous approaches unnecessary). Planar graph drawings have also been extensively studied in the presence of bends. Here, a fundamental result is by Tamassia~\cite{DBLP:journals/siamcomp/Tamassia87} in the context of \emph{orthogonal} graph drawings, i.e., drawings in which edges are axis-aligned polylines. In his seminal paper, Tamassia suggested an approach, called \emph{topology-shape-metrics}, to minimize the number of bends of degree-$4$ plane graphs using~flows. For a complete introduction, see~\cite{DBLP:books/ph/BattistaETT99}.

When the input graph is non-planar, however, the available approaches that yield aesthetically pleasing drawings are significantly fewer. The main obstacle here is that the presence of edge-crossings negatively affects the drawing's quality~\cite{DBLP:journals/iwc/Purchase00} and, on the other hand, their minimization turns out to be a computationally difficult problem~\cite{doi:10.1137/0604033}. In an attempt to overcome these issues, a decade ago, Huang et al.~\cite{DBLP:journals/vlc/HuangEH14} made a crucial observation that gave rise to a new line of research (currently recognized under the term ``beyond planarity''~\cite{DBLP:books/sp/20/HT2020}): edge crossings do not negatively affect the quality of the drawing too much (and hence the human's ability to read and interpret it), if the angles formed at the crossing points are large. Thus, the focus moved naturally to non-planar graphs and their properties, when different restrictions on the type of edge-crossings are imposed; see~\cite{DBLP:journals/csur/DidimoLM19} for an overview.

Among the many different classes of graphs studied as part of this emerging line of research, one of the most studied ones is the class of \emph{right-angle-crossing graphs} (or \emph{RAC graphs}, for short); see \cite{DBLP:books/sp/20/Didimo20} for a survey. These graphs were introduced by Didimo, Eades and Liotta~\cite{DBLP:conf/wads/DidimoEL09,DBLP:journals/tcs/DidimoEL11} back in 2009 as those admitting straight-line drawings in which the angles formed at the crossings are all $90^\circ$. Most notably, these graphs are optimal in terms of the crossing angles, which makes them more readable according to the observation by Huang et al.~\cite{DBLP:journals/vlc/HuangEH14}; moreover, RAC drawings form a natural generalization of orthogonal graph drawings~\cite{DBLP:journals/siamcomp/Tamassia87}, as any crossing between two axis-aligned polylines trivially yields $90^\circ$ angles. 

In the same work~\cite{DBLP:conf/wads/DidimoEL09,DBLP:journals/tcs/DidimoEL11}, Didimo, Eades and Liotta proved that every $n$-vertex RAC graph is sparse, as it can contain at most $4n-10$ edges, while in a follow-up work~\cite{DBLP:journals/ipl/DidimoEL10} they observed that not all degree-$4$ graphs are RAC. This gives rise to the following question which has also been independently posed in several subsequent works (see e.g.,~\cite{DBLP:journals/jgaa/AngeliniCDFBKS11}, \cite[Problem $6$]{Didimo2013}, \cite[Problem 9.5]{DBLP:books/sp/20/Didimo20}, \cite[Problem $8$]{DBLP:journals/csur/DidimoLM19}) and arguably forms the most intriguing open problem in the area, as it remains unanswered since more than one decade.

\begin{question}\label{q:degree-3-rac}
Does every graph with degree at most 3 admit a straight-line RAC drawing?
\end{question}

The most relevant result that is known stems from the related problem of simultaneously embedding two or more graphs on the Euclidean plane, such that the crossings between different graphs form $90^\circ$ angles. In this setting, Argyriou et al.~\cite{DBLP:journals/jgaa/ArgyriouBKS13} showed that a cycle and a matching always admit such an embedding, which implies that every \emph{Hamiltonian} degree-$3$ graph is RAC. 

Finally, note that recognizing RAC graphs is hard in the existential theory of the reals~\cite{DBLP:conf/gd/000121}, which also implies that RAC drawings may require double-exponential area, in contrast to the quadratic area requirement for planar graphs~\cite{DBLP:conf/stoc/FraysseixPP88}.

RAC graphs have also been studied by relaxing the requirement that the edges are straight-line segments, giving rise to the class of \emph{$k$-bend} RAC graphs (see, e.g, \cite{DBLP:journals/tcs/AngeliniBFK20,DBLP:journals/comgeo/ArikushiFKMT12,DBLP:journals/tcs/BekosDLMM17,DBLP:journals/comgeo/ChaplickLWZ19,DBLP:journals/mst/GiacomoDLM11}), i.e., those admitting drawings with at most $k$ bends per edge and crossings at $90^\circ$ angles. It is known that every degree-$3$ graph is $1$-bend RAC and every degree-$6$ graph is $2$-bend RAC~\cite{DBLP:journals/jgaa/AngeliniCDFBKS11}. While the flexibility guaranteed by the presence of one or two bends on each edge is not enough to obtain a RAC drawing for every graph (in fact, $1$- and $2$-bend RAC graphs with $n$ vertices have at most $5.5n-O(1)$ and $72n-O(1)$ edges, respectively~\cite{DBLP:journals/tcs/AngeliniBFK20,DBLP:journals/comgeo/ArikushiFKMT12}), it is known that every graph is $3$-bend RAC~\cite{DBLP:journals/mst/GiacomoDLM11} and fits on a grid of cubic size~\cite{DBLP:conf/esa/Forster020}.

\subparagraph{Our contribution.}
We provide several improvements to the state of the art concerning~RAC graphs with low degree. In particular, we make an important step towards answering~\cref{q:degree-3-rac} by proving that $3$-edge-colorable degree-$3$ graphs are RAC (\cref{thm:deg-3-col}). This result applies to Hamiltonian $3$-regular graphs, to bipartite $3$-regular graphs and, with some minor modifications to our approach, to all Hamiltonian degree-$3$ graphs, thus extending the result in~\cite{DBLP:journals/jgaa/ArgyriouBKS13}.  As a further step towards answering \cref{q:degree-3-rac}, we prove that bridgeless $3$-regular graphs with oddness at most $2$ are RAC (\cref{thm:oddness-2}). If their oddness is $k$, we provide an algorithm to construct a $1$-bend RAC drawing where at most $k$ edges have a bend(\cref{thm:oddness-k}).

We then focus on RAC drawings with one or two bends per edge. Namely, we prove that all degree-$4$ graphs admit $1$-bend RAC drawings and all $7$-edge-colorable degree-$7$ graphs admit $2$-bend RAC drawings (\cref{thm:degree-4-one,thm:7colorable}), which form non-trivial improvements over the state of~the~art, as the existence of such drawings was previously known only for degree-$3$ and degree-$6$ graphs~\cite{DBLP:journals/jgaa/AngeliniCDFBKS11}.

\section{Preliminaries}
Let $G = (V,E)$ be a graph. W.l.o.g. we assume that $G$ is connected, as otherwise we apply our drawing algorithms to each component of $G$ separately. 
$G$ is called \emph{degree-$k$} if the maximum degree of $G$ is $k$. 
$G$ is called \emph{$k$-regular} if the degree of each vertex of $G$ is exactly $k$.
A $2$-factor of an undirected graph $G=(V,E)$ is a spanning subgraph of $G$ consisting of vertex disjoint cycles. 
Let $F$ be a $2$-factor of $G$ and let $\prec$ be a total order of the vertices such that the vertices of each cycle $C \in F$ appear consecutive in $\prec$ according to some traversal of $C$. In other words, every two vertices that are adjacent in $C$ are consecutive in $\prec$ except for two particular vertices, which are the first and the last vertices of $C$ in $\prec$. We call the edge between these two vertices the \emph{closing edge} of $C$.
By definition, $\prec$ also induces a total order of the cycles of $F$. Let $\{u,v\}$ be an edge in $E \setminus F$ and let $C$ and $C'$ be the cycles of $F$ that contain $u$ and $v$, respectively.
If $C = C'$, then $\{u,v\}$ is a \emph{chord} of $C$. Otherwise, $C \neq C'$. If $u \prec v$, $\{u,v\}$ is called a \emph{forward edge} of $u$ and a \emph{backward edge} of $v$.
The following theorem provides a tool to partition the edges of a bounded degree graph into $2$-factors \cite{10.1007/BF02392606}.
\begin{theorem}
[Eades, Symvonis, Whitesides \cite{DBLP:journals/dam/EadesSW00}] \label{thm:2-factors} Let $G=(V,E)$ be an n-vertex undirected graph of degree $\Delta$ and let $d = \lceil \Delta / 2 \rceil$. Then, there exists a directed multi-graph $G' = (V,E')$ with the following properties:
\begin{enumerate}
    \item each vertex of $G'$ has indegree $d$ and outdegree $d$;
    \item $G$ is a subgraph of the underlying undirected graph of $G'$; and
    \item the edges of $G'$ can be partitioned into $d$ edge-disjoint directed $2$-factors.
\end{enumerate}

Furthermore, the directed graph $G'$ and the $d$ $2$-factors can be computed in $\mathcal{O}(\Delta^2 n)$ time.
\end{theorem}

Let $\Gamma$ be a polyline drawing of $G$ such that the vertices and edge-bends lie on grid points. The area of $\Gamma$ is determined by the smallest enclosing rectangle. Let $\{u,v\}$ be an edge in $\Gamma$. We say that $\{u,v\}$ is using an \emph{orthogonal port} at $u$ if the edge-segment $s_u$ of $\{u,v\}$ that is incident to $u$ is either horizontal or vertical; otherwise it is using an \emph{oblique port} at $u$.
We denote the orthogonal ports at $u$ by $N$, $E$, $S$ and $W$, if $s_u$ is above, to the right, below or to the left of $u$, respectively. If no edge is using a specific orthogonal port, we say that this port is \emph{free}.

\section{RAC drawings of $\mathbf{3}$-edge-colorable degree-$\mathbf{3}$ graphs}
\label{sec:degree-3}

In this section, we prove that $3$-edge-colorable degree-$3$ graphs admit RAC drawings of quadratic area, which can be computed in linear time assuming that the edge coloring is given (testing the existence of such a coloring is NP-complete even for $3$-regular graphs~\cite{doi:10.1137/0210055}).

\begin{theorem} \label{thm:deg-3-col}
Given a $3$-edge-colorable degree-$3$ graph $G$ with $n$ vertices and a $3$-edge-coloring of $G$, it is possible to compute in $O(n)$ time a RAC drawing of $G$ with $O(n^2)$ area.
\end{theorem}

\noindent We assume w.l.o.g.\ that $G$ does not contain degree-$1$ vertices, as otherwise we can replace each such vertex with a $3$-cycle while maintaining the $3$-edge-colorability of the graph and without asymptotically increasing the size of the graph.
Since $G$ is $3$-edge-colorable, it can be decomposed into three matchings $M_1$, $M_2$ and $M_3$. In the produced RAC drawing, the edges in $M_1$ will be drawn horizontal, those in $M_3$ vertical, while those in $M_2$ will be crossing-free, not maintaining a particular slope. Let $H_y$ and $H_x$ be two subgraphs of $G$ induced by $M_1 \cup M_2$ and $ M_2 \cup M_3$, respectively. Since every vertex of $G$ has at least two incident edges, which belong to different matchings, each of $H_y$ and $H_x$ spans all vertices of $G$. Further, any connected component in $H_y$ or $H_x$ is either a path or an even-length~cycle, as both $H_y$ and $H_x$ are degree-$2$ graphs alternating between edges of different matchings.

\begin{figure}[t]
    \centering
    \begin{subfigure}[b]{.38\textwidth}
    \centering
    \includegraphics[scale=0.6,page=1]{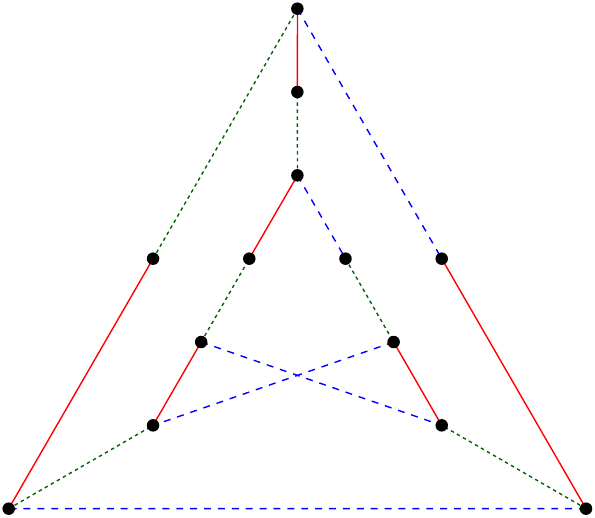}
    \subcaption{$G$}
    \end{subfigure}
    \hfil
    \begin{subfigure}[b]{.31\textwidth}
    \centering
    \includegraphics[scale=0.6,page=2]{images/example-graph-alternative.pdf}
    \subcaption{$c_1,\dots, c_5$}
    \end{subfigure}
    \hfil
    \begin{subfigure}[b]{.27\textwidth}
    \centering
    \includegraphics[scale=0.8,page=3]{images/example-graph-alternative.pdf}
    \subcaption{$\mathcal{H}$}
    \end{subfigure}
    \caption{(a) A non-planar, non-Hamiltonian, 3-edge-colorable degree-$3$ example graph $G$. The matching $M_1$ is drawn with blue dashed lines, $M_2$ with red solid lines and $M_3$ with green dotted lines. (b) The components $c_1$, $c_2$ and $c_3$ of the subgraph $H_y$ induced by $M_1 \cup M_2$ (shaded in blue) and the components $c_4$ and $c_5$ of $H_x$ induced by $M_2 \cup M_3$ (shaded in green). (c) The auxiliary graph $\mathcal{H}$ in which the components of $H_y$ and $H_x$ sharing at least one vertex are connected by an edge. In this example, the BFS traversal of the components of $\mathcal{H}$ is $c_1, c_2, c_3, c_4, c_5$.}
    \label{fig:example-graph-deg-3}
\end{figure}

We define an auxiliary bipartite graph $\mathcal{H}$, whose first (second) part has a vertex for each connected component in $H_y$ ($H_x$), and there is an edge between two vertices if and only if the corresponding components share at least one vertex; see \cref{fig:example-graph-deg-3}.

\begin{prop}
The auxiliary graph $\mathcal{H}$ is connected.
\end{prop}
\begin{proof}
Suppose for a contradiction that $\mathcal{H}$ is not connected. Let $v_c$ and $v_{c'}$ be two vertices of $\mathcal{H}$ that are in different connected components of $\mathcal{H}$. By definition of $\mathcal{H}$, $v_c$ and $v_{c'}$ correspond to connected components $c$ and $c'$, respectively, of $H_y$ or $H_x$. W.l.o.g.\ assume that $c$ belongs to $H_y$. Let $u_0$ and $u_k$ be two vertices of $G$ that belong to $c$ and $c'$, respectively. Since $G$ is connected, there is a path $P =(u_0,u_1,u_2,\dots u_k)$ between $u_0$ and $u_k$ in $G$, such that no two consecutive edges in $P$ belong to the same matching. Let $(u_{i},u_{i+1})$ be the first edge of $P$ from $u_0$ to $u_k$ that belongs to $M_3$, which exists since $c \neq c'$. By construction, this implies that $u_i$ belongs to $c$ and to another component $c^*$ of $H_x$. 
By definition of $\mathcal{H}$, $v_c$ and $v_{c^*}$ are connected in $\mathcal{H}$, where $v_{c^*}$ is the corresponding vertex of $c^*$ in $\mathcal{H}$. Repeating this argument until $u_k$ is reached yields a path in $\mathcal{H}$ from $v_c$ to $v_{c'}$, a contradiction. \end{proof}

We now define two total orders $\prec_y$ and $\prec_x$ of the vertices of $G$, which will then be used to assign their $y$- and $x$-coordinates, respectively, in the final RAC drawing of $G$.
Since we seek to draw the edges of $M_1$ ($M_3$) horizontal (vertical), we require that the endvertices of any edge in $M_1$ ($M_3$) are consecutive in $\prec_y$ ($\prec_x$, respectively). Moreover, for the edges of $M_2$, we guarantee some properties that will allow us to draw them without crossings. 

To construct $\prec_y$ and $\prec_x$, we process the components of $H_y$ and $H_x$ according to a certain BFS traversal of $\mathcal{H}$ and for each visited component of $H_y$ ($H_x$), we append all its vertices to $\prec_y$ ($\prec_x$) in a certain order. 

To select the first vertex of the BFS traversal of $H$, we consider a vertex $u$ of $G$ belonging to two components $c$ and $c'$ of $H_y$ and $H_x$, respectively, such that $u$ is the endpoint of $c$ if $c$ is a path; if $c$ is a cycle, we do not impose any constraints on the choice of $u$. We refer to vertex $u$ as the \emph{origin vertex} of $G$. Also, let $v_c$ and $v_{c'}$ be the vertices of $\mathcal{H}$ corresponding to $c$ and $c'$, respectively. By definition of $\mathcal{H}$, $v_c$ and $v_{c'}$ are adjacent in $\mathcal{H}$. We start our BFS traversal of $\mathcal{H}$ at $v_c$ and then we move to $v_{c'}$ in the second step (note that this choice is not needed for the definition of $\prec_y$ and $\prec_x$, but it guarantees a structural property that will be useful later). From this point on, we continue the BFS traversal to the remaining vertices of $\mathcal{H}$ without further~restrictions.
In the following, we describe how to process the components of $H_y$ and $H_x$ in order to guarantee an important property (see~\cref{prop:consecutive})

Let $c$ be the component of $H_y$ or $H_x$ corresponding to the currently visited vertex in the traversal of $\mathcal{H}$. Since $\mathcal{H}$ is bipartite, no other component of $H_y$ ($H_x$) shares a vertex with $c$, if $c$ belongs to $H_y$ ($H_x$). Hence, no vertex of $c$ already appears in $\prec_y$ ($\prec_x$).

If $c$ is a path, then we append the vertices of $c$ to $\prec_y$ or $\prec_x$ in the natural order defined by a walk from one of its endvertices to the other. Note that if $c$ is the first component in the BFS traversal of $\mathcal{H}$, one of these endvertices is by definition the origin vertex of $G$, which we choose to start the walk. Hence, in the following we focus on the case that $c$ is a cycle. In this case, the vertices of $c$ will also be appended to $\prec_y$ or $\prec_x$ in the natural order defined by some specific walk of $c$, such that the so-called \emph{closing edge} connecting the first and the last vertex of $c$ in this order belongs to $M_2$. Note that an edge might be closing in both orders $\prec_y$ and $\prec_x$.

Suppose first that $c \in H_y$. If $c$ is the first component in the BFS traversal of $\mathcal{H}$, then we append the vertices of $c$ to $\prec_y$ in the order that they appear in the cyclic walk of $c$ starting from the origin vertex of $G$ and following the edge of $M_1$ incident to it. Otherwise, let $v$ be the first vertex of $c$ in $\prec_x$, which is well defined since there is at least one vertex of $c$ that is part of $\prec_x$, namely, the one that is shared with its parent.
We append the vertices of $c$ to $\prec_y$ in the order that they appear in the cyclic walk of $c$ starting from $v$ and following the edge of $M_1$ incident to $v$. Hence, $v$ is the first vertex of $c$ in both $\prec_y$ and $\prec_x$. In both cases, it follows that the closing edge of $c$ belongs to $M_2$. 

Suppose now that $c \in H_x$, which implies that $c$ is not the first component in the BFS traversal of $\mathcal{H}$. Let $v$ be the first vertex of $c$ in $\prec_y$, which is again well defined since there is at least one vertex of $c$ that is part of $\prec_y$. We append the vertices of $c$ to $\prec_x$ in the inverse order that they appear in the cyclic walk of $c$ starting from $v$ and following the edge $(v, w)$ of $M_3$ incident to $v$ (or equivalently, in the order they appear in the cyclic walk of $c$ starting from the neighbor of $v$ different from $w$ and ending at $v$). Hence, $v$ is the first vertex of $c$ in $\prec_y$ and the last vertex of $c$ in $\prec_x$. Also in this case, the closing edge of $c$ belongs to $M_2$. See \cref{fig:example-ordering} for an illustration. Note that the closing edge of a component $c$ is contained inside the parent component of $c$ is the BFS traversal.
Moreover, by construction, the following property holds.

\begin{figure}[t]
    \centering
    \begin{subfigure}[b]{.31\textwidth}
    \centering
    \includegraphics[scale=0.6,page=1]{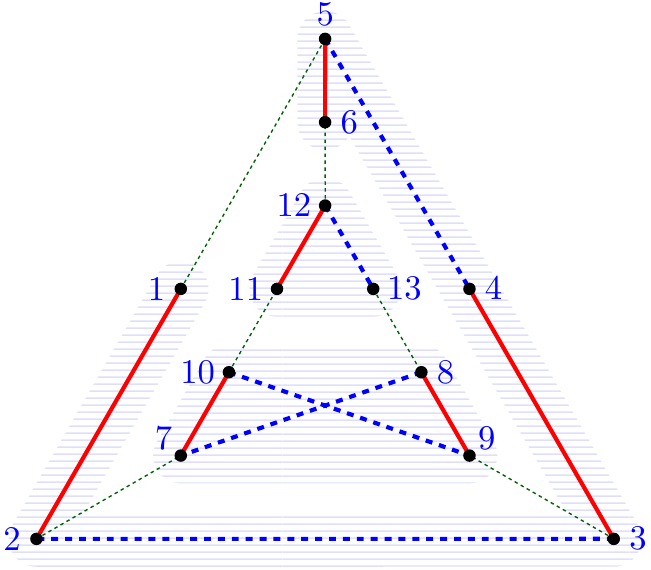}
    \subcaption{}
    \end{subfigure}
    \hfil
    \begin{subfigure}[b]{.31\textwidth}
    \centering
    \includegraphics[scale=0.6,page=2]{images/example-ordering.pdf}
    \subcaption{}
    \end{subfigure}
    \begin{subfigure}[b]{.31\textwidth}
    \centering
    \includegraphics[scale=0.5,page=3]{images/example-ordering.pdf}
    \subcaption{}
    \end{subfigure}
    \caption{The total orders (a)~$\prec_y$ for $H_y$ that consists of the blue and red edges, and (b)~$\prec_x$ for $H_x$ that consists of the green and red edges. The final drawing of $G$ is shown in (c).}
    \label{fig:example-ordering}
\end{figure}

\begin{prop}\label{prop:consecutive}
The endvertices of any edge in $M_1$ ($M_3$) are consecutive in $\prec_y$ ($\prec_x$). The endvertices of any edge in $M_2$ are consecutive in $\prec_y$ ($\prec_x$) unless this edge is a closing edge in a component of $H_y$ (of $H_x$).
\end{prop}

\medskip\noindent\textbf{Computing the vertex coordinates.} 
We use $\prec_y$ and $\prec_x$ to specify the $y$- and the $x$-coordinates of the vertices, respectively. To do so, we iterate through $\prec_y$ and set the $y$-coordinate of its first vertex to $1$. Let $v$ be the next vertex in the iteration and let $u$ be its predecessor in $\prec_y$.
Assume that the $y$-coordinate of $u$ is $i$. If $(u,v) \in M_1$, we set the same $y$-coordinate $i$ to $v$. Otherwise, either $u$ and $v$ belong to two different components of $H_y$ or $(u,v) \in M_2$ and we set the $y$-coordinate $i+1$ to $v$. Similarly, we iterate through $\prec_x$ and set the $x$-coordinate of its first vertex to $1$. Let $v$ be the next vertex in the iteration and let $u$ be its predecessor in $\prec_x$.
Assume that the $x$-coordinate of $u$ is $i$. If $(u,v) \in M_3$, we set the $x$-coordinate of $v$ to $i$. Otherwise, either $u$ and $v$ belong to two different components of $H_x$ or $(u,v) \in M_2$ and we set the $x$-coordinate $i+1$ to $v$. Hence, no two vertices share the same $x$- and $y$-coordinates.
We next show that the computed vertex coordinates induce a straight-line RAC drawing $\Gamma$ of $G$ with the possible exception of the edge of $M_2$ incident to the origin vertex of $G$, since this edge would be a closing edge for both $c$ and $c'$ and hence by \cref{prop:consecutive} its endpoints would be consecutive in neither $\prec_y$ nor $\prec_x$. If this edge exists, we denote it by $e^*$, while the graph $G \setminus e^*$ and its drawing in $\Gamma$ are denoted by $G^*$ and $\Gamma^*$, respectively.

\begin{lemma}\label{lem:construction}
Let $e$ be an edge of $G^*$. Then, $e$ is drawn horizontally in $\Gamma^*$ if $e \in M_1$; vertically in $\Gamma^*$ if $e \in M_3$ and crossing-free in $\Gamma^*$ if $e \in M_2$.
\end{lemma}
\begin{proof}
If $e \in M_1$ or $e \in M_3$, the statement follows from \cref{prop:consecutive} and the computed vertex coordinates. Hence, let $e=(u,v)$ be an edge of $M_2$ and let $c_y \in H_y$ and $c_x \in H_x $ be the two components of $\mathcal{H}$ containing $e$.
Suppose to the contrary that there is an edge $e' = (u',v')$ crossing $e$. If $e' \in M_1$, then both $u'$ and $v'$ belong to the same component $c_y' \in H_y$. If $c_y' \neq c_y$, then $e'$ cannot cross $e$ as the vertices of $c_y$ and $c_y'$ span different intervals of $y$-coordinates, hence $e'$ belongs to $c_y$. Similarly, if $e' \in M_3$, then $e'$ belongs to $c_x$. Finally, if $e' \in M_2$, then $u'$ and $v'$ belong to the same component in both $H_y$ and $H_x$; thus $e'$ belongs to both~$c_y$~and~$c_x$. 

\begin{enumerate}
    \item Edge $e$ is a closing edge for neither $c_y$ nor $c_x$: The vertices $u$ and $v$ are consecutive in both $\prec_y$ and $\prec_x$ by  \cref{prop:consecutive}.
     Hence, both their $x$- and $y$-coordinate differ by exactly one by construction. Since all vertices have integer coordinates, no horizontal or vertical edge is crossing $e$, hence $e' \in M_2$.
     Observe that since the $y$- (the $x$-) coordinate of the vertices in $c_y$ (in $c_x$) are non-decreasing along the walk defining its order in $\prec_y$ (in $\prec_x$), no crossing between $e$ and $e'$ can occur if $e'$ is not a closing edge of $c_y$ or $c_x$, which will be covered in the next cases (by swapping the roles of $e$ and $e'$).
    \item Edge $e$ is a closing edge for $c_y$ but not for $c_x$: Note that $c_y$ is not the first component in the BFS traversal, since the closing edge of this component is $e^*$. Further, $u$ and $v$ are consecutive in $\prec_x$, but not in $\prec_y$. We assume that $u$ directly precedes $v$ in $\prec_x$, which implies that $u$ is the first vertex of $c_y$ in $\prec_y$, while $v$ is the last.
    It follows that their $x$-coordinates differ by exactly one, hence $e'$ cannot belong to $M_3$. If $e'$ belongs to $M_1$, then one of $u'$ or $v'$, say $u'$, has $x$-coordinate smaller or equal to the one of $u$. Since $u$ and $v$ are consecutive in $\prec_x$, we have necessarily that $u'$ precedes $u$ in $\prec_x$, which is a contradiction to the choice of $u$ since both $u'$ and $v'$ belong to $c_y$. In fact, $u$ was chosen as the starting point of the walk, when considering $c_y$, as the first vertex of $c_y$ in $\prec_x$, hence $e' \in M_2$. Since both endpoints of $e'$ belong to both $c_y$ and $c_x$, then one of $u'$ or $v'$, say $u'$, has $x$-coordinate smaller or equal to the one of $u$. Since $u$ and $v$ are consecutive in $\prec_x$, we have necessarily that $u' \prec_x u$, which is a contradiction to the choice of $u$ since both $u'$ and $v'$ belong to $c_y$.
    
    \item Edge $e$ is a closing edge for $c_x$ but not for $c_y$: This case is analogous to the previous one.
    \item\label{case:closing}Edge $e$ is a closing edge for both $c_y$ and $c_x$: Observe that
    neither $c_y$ nor $c_x$ is the first component in the BFS traversal of $\mathcal{H}$, since $e^*$ is the closing edge of this component, which is not part of $G^*$. Recall that by definition, the vertices $v_{c_y}$ and $v_{c_x}$ corresponding to $c_y$ and $c_x$ in $\mathcal{H}$ are adjacent. Assume that $c_y$ is visited before $c_x$ in the BFS traversal; the other case is symmetric. By our construction rule, when considering $c_y$, we started the walk from the vertex $u$ that is the first vertex of $c_y$ in $\prec_x$, which means $u$ also belongs to a component $c_x'$ of $H_x$. Since $c_y$ is a cycle, $u$ is incident to an edge in $M_2$, which then also belongs to $c_x'$. Clearly, the edge of $M_2$ incident to $u$ is the closing edge of $c_y$ and contained in $c_x'$, which implies that the edge does not belong to $c_x$, hence this case does not occur in $G^*$.\qedhere
    \end{enumerate}
\end{proof}

By the last case of~\cref{lem:construction}, it follows that if the edge $e^*$ exists, then it is the only closing edge of two components, which is summarized in the following corollary.
\begin{corollary}\label{cor:single-closing}
There is at most one edge in $M_2$ that is a closing edge for two components.
\end{corollary}
We now describe how to add the edge $e^*$ to $\Gamma^*$ if such an edge exists to obtain the final drawing $\Gamma$.
Let $u$ and $v$ be the endvertices of $e^*$ with $u$ being the origin vertex of $G$. By construction, $u$ and $v$ are in the first two components $c$ and $c'$ of the BFS traversal of $\mathcal{H}$. Since $u$ is the first vertex in $\prec_y$, its $y$-coordinate is $1$, i.e., $u$ is the bottommost vertex of $\Gamma^*$. Also, since $u$ is the first vertex of $c'$ in $\prec_y$, it is incident to the closing edge of $c'$ and $c$ by definition, in particular, this edge is $e^*$. Note that this implies that the $x$-coordinate of $v$ is $1$, so $v$ is the leftmost vertex of $\Gamma^*$ and the first vertex in $\prec_x$.
This ensures that $v$ can be moved to the left and $u$ to the bottom in order to draw the edge $e^*$ crossing-free. In particular, moving $v$ by $n$ units to the left and $u$ by $n$ units to the bottom we can guarantee that $e^*$ does not intersect the first quadrant $\mathbb{R}_+^2$, while by construction any other edge (not incident to $u$ or $v$) lies in $\mathbb{R}_+^2$. Since $e^*$ is the only edge of $M_2$ incident to $u$ and $v$, it remains to consider the edges of $M_1$ and $M_3$ incident to $u$ or $v$. Observe that the edge of $M_1$ incident to $v$ remains horizontal, while the edge of $M_3$ incident to $u$ remains vertical. Finally, the edge of $M_1$ incident to $u$ is crossing free in $\Gamma^*$, since there is no vertex below it, hence it remains crossing-free after moving $u$ to the bottom. Similarly, the edge of $M_3$ incident to $v$ is crossing free in $\Gamma^*$, since there is not vertex to the left of it, hence it remains crossing-free after moving $v$ to the left. Together with \cref{lem:construction} we obtain that $\Gamma$ is a RAC drawing of $G$. 
We complete the proof of \cref{thm:deg-3-col} by discussing the time complexity and the required area.
We construct the components of $H_y$ and $H_x$ based on the given edges-coloring using BFS in $\mathcal{O}(n)$ time. 
To define $\prec_y$ and $\prec_x$, we choose the origin vertex $u$ and the components $c$ and $c'$ for the start of the BFS of $\cal H$ in linear time. We then traverse every edge of $G$ at most twice. Hence, this step takes $\mathcal{O}(n)$ time in total. Assigning the vertex coordinates, by first iterating through $\prec_y$ and $\prec_x$ and then possibly moving the end-vertices of $e^*$, clearly takes $\mathcal{O}(n)$ time again, hence the drawing can be computed in linear time.

For the area, we observe that the initial $x$- and $y$-coordinates for all the vertices range between $1$ and $n$. Since we possibly move the origin vertex $u$ and its $M_2$ neighbor $v$ by $n$ units each, the drawing 
area is at most $2n \times 2n$.

We conclude this section by mentioning two results that form generalizations of our approach of \cref{thm:deg-3-col}. In this regard, we need the notion of oddness of a bridgeless $3$-regular graph, which is defined as the minimum number of odd cycles in any possible $2$-factor of~it. \cref{thm:oddness-2} is limited to oddness-2, while \cref{thm:oddness-k} provides an upper bound on the number of edges requiring one bend that is linear in the oddness; \mbox{their proofs are in the appendix.}

\begin{toappendix}
This section is devoted to the proofs of Theorems~\ref{thm:oddness-2} and~\ref{thm:oddness-k}. We start with the formal definition of oddness, which is limited to given in the following definition.

\begin{definition}[Huck and Kochol~\cite{huck1995five}]\label{def:oddness}
The oddness of a $3$-regular graph is the minimum number of odd cycles in a $2$-factor of it.
\end{definition}

\noindent Before we proceed with the formal proof of Theorems~\ref{thm:oddness-2} and~\ref{thm:oddness-k}, we introduce two auxiliary properties that hold for every graph with oddness $k \geq 2$. 

\begin{prop}\label{prop:k-m4}
Given a bridgeless $3$-regular with oddness $k \geq 2$, it is possible to decompose $G$ into four matchings $M_1$, $M_2$, $M_3$ and $M_4$, such that 
\begin{enumerate*}[label={(\roman*)}, ref=(\roman*)]
\item one of them is perfect, say $M_1$, \label{prp:k-m4-1}
\item another one, say $M_4$, has exactly $k$~edges $e_1,\ldots,e_k$, and \label{prp:k-m4-2}
\item $M_2\cup M_3 \cup M_4$ yield a $2$-factor having only $k$ odd cycles $c_1,\ldots,c_k$, such that $e_i$ can be arbitrarily chosen from the odd cycle $c_i$.\label{prp:k-m4-3}
\end{enumerate*}
\end{prop}
\begin{proof}
Since $G$ has oddness $k$, there exists a perfect matching $M_1$ in $G$ such that the $2$-factor $G \setminus M_1$ has exactly $k$ odd cycles $c_1,\ldots,c_k$. This guarantees Properties~\ref{prp:k-m4-1} and~\ref{prp:k-m4-2}.  This $2$-factor is decomposable into three matchings $M_2$, $M_3$ and $M_4$, such that each odd cycle $c_i$ contains exactly one edge $e_i$ in $M_4$, which can be arbitrarily chosen. This proves Property~\ref{prp:k-m4-3}. Thus, the overall cardinality of $M_4$ is exactly $k$ and the proof is concluded.
\end{proof}

\noindent For the special case in which the oddness is $2$ we can further strengthen \cref{prop:k-m4} as follows.

\begin{prop}\label{prop:2-m4}
Given a bridgeless $3$-regular graph with oddness $2$, it is possible to decompose $G$ into four matchings $M_1$, $M_2$, $M_3$ and $M_4$, such that
\begin{enumerate*}[label={(\roman*)}, ref=(\roman*)]
\item\label{prp:1}$M_1$ is a perfect matching,
\item\label{prp:2}$M_2$, $M_3$ and $M_4$ yield a $2$-factor having only two odd cycles $c$ and $c'$, and
\item\label{prp:3} $c$ and $c'$ contain two vertices $u$ and $u'$, such that there is a path $P(u,u')$ between $u$ and $u'$ that consists exclusively of edges in $M_1 \cup M_2$ and $P(u,u')$ does not contain any vertex of $c$ or $c'$ other than $u$ and $u'$.
\end{enumerate*}
\end{prop}
\begin{proof}
By \cref{prop:k-m4}, it follows that $G$ can be decomposed into four matchings $M_1$, $M_2$, $M_3$ and $M_4$, such that $M_1$ is a perfect matching while $M_4$ contains exactly two edges $(a,b)$ and $(a',b')$ belonging to the two odd cycles $c$ and $c'$, respectively. Consider $c \setminus (a,b)$. By definition, the edges of $M_2$ and $M_3$ are alternating when walking along $c \setminus (a,b)$ from $a$ to $b$ and since the walk has even length it follows that $a$ is incident to an edge of $M_3$ and $b$ is incident to an edge of $M_2$ or vice-versa; assume w.l.o.g.\ the former holds. Symmetrically, we can prove that $a'$ and $b'$ are incident to two edges in $M_3$ and $M_2$, respectively. Hence, any vertex in $G$ besides $a$ and $a'$ has degree exactly $2$ in $M_1 \cup M_2$. Since $a$ and $a'$ have degree $1$ in $M_1 \cup M_2$, it follows that there is a path $Q$ connecting $a$ and $a'$ in $M_1 \cup M_2$. Let $u$ ($u'$) be the last (first) occurrence of a vertex of $c$ ($c'$) along $Q$ from $a$ to $a'$. Note that $u$ and $u'$ may possibly coincide with $a$ and $a'$ if $(a,a') \in M_1$. It follows that the subpath $P$ of $Q$ from $u$ to $u'$ contains only two vertices of $c$ and $c'$, namely, its endpoints $u$ to $u'$. To complete the proof, we locally change the assignment of the edges of $c$ and $c'$ to $M_2$, $M_3$ and $M_4$ such that $u$ and $u'$ will be incident to an edge of $M_3$ and an edge of $M_2$ each. This way we  guarantee that $P$ is actually a path in $M_1 \cup M_2$ completing the proof.
\end{proof}
\end{toappendix}

\begin{theorem2rep}\apxmark\label{thm:oddness-2}
Every bridgeless $3$-regular graph with oddness $2$ admits a RAC drawing in quadratic area which can be computed in subquadratic time.
\end{theorem2rep}
\begin{proof}
To prove the theorem, we make use of \cref{prop:2-m4}, which allows us to assume that $G$ is decomposed into four matchings $M_1$, $M_2$, $M_3$ and $M_4$, such that $M_1$ is perfect and $M_4$ contains only two edges. Also, the union of $M_2$, $M_3$ and $M_4$ is a $2$-factor containing two odd cycles, which we denote by $c$ and $c'$. Based on path $P$ from $u \in c$ to $u'\in c'$, we construct a colorable graph $G'$ as follows. Let $\{u, w\} \in M_4$ and $\{u', w'\} \in M_4$. We substitute these edges by two paths $u \rightarrow v \rightarrow w$ and $u' \rightarrow v' \rightarrow w'$ that contain dummy vertices $v$ and $v'$ respectively. The edges $\{u, v\}$ and $\{u', v'\}$ are added to $M_2$ while the edges $\{v, w\}$ and $\{v', w'\}$ are added to $M_3$.

It follows that $G'$ is a $3$-edge-colorable degree-$3$ graph. Thus, $G'$ is then drawn as described in \cref{thm:deg-3-col} by identifying $H_y$ to be the subgraph induced by $M_1 \cup M_2$ and $H_x$ to be the subgraph induced by $M_2 \cup M_3$. Also, we set $u$ to be the origin vertex of the BFS traversal, which places the cycle $c$ at the left side of the drawing, since the vertices of $c$ are the first vertices considered in $\prec_y$. Afterwards, we reverse the internal order of $c'$ in $\prec_x$ (which is equivalent to flipping the cycle in the drawing along the $y$-axis) and move $c'$ to the right side of the drawing, such that $\forall c \in H_x$ it holds that all vertices of $c$ are before the vertices of $c'$ with respect to $\prec_x$. Let $\Gamma'$ be the resulting drawing. We prove in the following lemma that $\Gamma'$ is a RAC drawing of graph $G'$.

\begin{lemma}
$\Gamma'$ is a RAC drawing of graph $G'$.
\end{lemma}
\begin{appproof}
Let $e$ be an edge that is incident to a vertex of $c'$. The following three cases can arise since moving $c'$ can affect vertical, horizontal and diagonal edges.
First assume that $e \in M_1$ is horizontal. Moving $c'$ along the $x$-axis does not change the $y$-alignment of the endvertices of $e$ and the edge stays horizontal. Next, assume that such an edge $e \in M_3$ is vertical. Then $e \in M_1$ is an internal edge of $c'$ which means that both endpoints of $e$ are part of $c'$. Thus moving $c'$ along the $x$-axis keeps both endpoints aligned such that $e$ stays vertical. In order for $\Gamma'$ to be a valid RAC drawing, the edges in $M_2$ are required to be crossing-free. If an edge $e \in M_2$ is not a closing edge in $H_y$ or $H_x$, then it is consecutive in both $\prec_y$ and $\prec_x$. This property is not altered by moving and flipping $c'$. Otherwise, $e$ is either a closing edge in $H_y$ or $H_x$ but not both, since $u$ was chosen as the origin vertex. If $e$ is a closing edge in $H_x$, then it was chosen with the minimal $\prec_y$ value. This property is still valid after moving along the $x$-axis and flipping along the $y$-axis, since the vertical positions of the vertices in $c'$ are not changed. If $e$ is a closing edge in $H_y$, it was chosen with the minimal $\prec_x$ value. After moving $c'$, $e$ could potentially be involved in crossing as $e$ is not the leftmost edge in its component $c$ in $H_y$ anymore. More precisely, it is to the right of every vertex in $c$ which is not in $c'$ and it is to the left of every vertex in $c$ which is also in $c'$. To avoid this, we flip $c'$, i.e., we mirror the drawing of $c'$ along the vertical axis. Thus, $e$ becomes the rightmost edge in $c$ and is therefore planar.
\end{appproof}

To complete the proof of \cref{thm:oddness-2}, we prove in the following that a valid RAC drawing for $G$ can be derived from the RAC drawing $\Gamma'$ of $G'$. We argue as follows. The inserted paths are contracted to retrieve $M_4$ by merging $v$ with $u$ at the coordinate $(x(u), y(u))$ and merging $v'$ with $u'$ at the coordinate $(x(v'), y(u'))$. This results in a diagonal drawing of the edge $e \in M_4$ incident to $u$ and the edge $e' \in M_2$ incident to $u'$. Since $u$ is the starting vertex for the BFS traversal, $P$ is positioned at the bottom of the drawing and is only connected to other $H_y$ components over vertical edges. $P$ can thus be moved downwards. Furthermore $e$ is connected to the leftmost vertex $w$ which can be moved to the left in order to guarantee that $e$ is crossing free. Similarly by the construction, $e'$ is incident to the rightmost vertex, which can be moved to the right to guarantee correctness of the algorithm. It remains to discuss the area and the time complexity. \cref{thm:deg-3-col} yields a drawing in quadratic area. It is sufficient to move the vertices $w$ and $w'$ by at most $n$ units to the left or right, respectively and the vertices of $P$ downwards by $n$ units in order to draw the edges $(u,w)$ and $(u',w')$, hence we maintain quadratic area. For the time complexity, we observe that the running time of \cref{prop:2-m4} is dominated by finding a perfect matching in a cubic bridgeless graph, which can be done in subquadratic time (see e.g. \cite{BIEDL2001110}), while the later operations only require linear time. Applying \cref{thm:deg-3-col} and post processing clearly takes linear time, as we have the $3$-edge coloring given, which concludes the proof.
\end{proof}

\begin{theorem2rep}\apxmark\label{thm:oddness-k}
Every bridgeless $3$-regular graph with oddness $k \geq 2$ admits a $1$-bend RAC drawing in quadratic area where at most $k$ edges require one bend.
\end{theorem2rep}
\begin{proof}
Let $G$ be a bridgeless $3$-regular graph with oddness $k\geq 2$. It follows that $G$ is $4$-edge-colorable. In order to apply the algorithm of \cref{thm:deg-3-col}, we transform $G$ to a $3$-edge-colorable degree-$3$ graph $G'$ by substituting each edge $(u, w) \in M_4$ (recall that $|M_4|$ is equivalent to the oddness of $G$ by \cref{prop:k-m4}) by a length-$2$ path $u \rightarrow v \rightarrow w$, where $v$ is a new vertex. Then, $G'$ is $3$-edge-colorable, since we have two different colors at $u$ and $w$ left such that the edges incident to $v$ can be properly colored. Afterwards, the algorithm of \cref{thm:deg-3-col} yields a straight-line drawing of $G'$. The 1-bend drawing of $G$ is obtained by converting each newly introduced $v$ to a bend for the original edge $(u,w)$.
\end{proof}

\section{$\mathbf{1}$-Bend RAC Drawings of Degree-$\mathbf{4}$ graphs}
\label{sec:degree-4}

In this section, we focus on degree-$4$ graphs and show that they admit $1$-bend RAC drawings.

\begin{theorem}\label{thm:degree-4-one}
Given a  degree-$4$ graph $G$ with $n$ vertices, it is possible to compute in $O(n)$ time a $1$-bend RAC drawing of $G$ with $O(n^2)$ area.
\end{theorem}

\begin{proof}
By \cref{thm:2-factors}, we augment $G$ into a directed $4$-regular multigraph $G'$ with edge disjoint $2$-factors $F_1$ and $F_2$.
Let $G_s$ be the graph obtained from $G'$ as follows. For each vertex $u$ of $G'$ with incident edges $(a_1,u),(u,b_1) \in F_1$ and $(a_2,u),(u,b_2) \in F_2$, we add two vertices $u_s$ and $u_t$ to $G_s$ that are incident to the following five edges: $u_s$ is incident to the two incoming edges of $u$, namely, $(a_1,u_s)$ and $(a_2,u_s)$, while $u_t$ is incident to $(u_t,b_1)$ and $(u_t,b_2)$. Finally, we add the edge $(u_s,u_t)$ to $G_s$, which we call \emph{split-edge}. 
 
By construction, $G_s$ is $3$-regular and $3$-edge colorable, since each vertex of it is incident to one edge of $F_1$, one edge of $F_2$ and one split-edge. By applying the algorithm of \cref{thm:deg-3-col} to $G_s$, we obtain a RAC drawing $\Gamma_s$ of $G_s$, such that the matching $M_2$ in the algorithm is the one consisting of all the split-edges.
To obtain a $1$-bend drawing for $G'$, it remains to merge the vertices $u_s$ and $u_t$ for every vertex $u$ in $G'$. We place $u$ at the position of $u_s$ in $\Gamma_s$. We draw each outgoing edge $(u,x)$ of $u$ as a polyline with a bend placed close to the position of $u_t$ in $\Gamma_s$ (the specific position will be discussed later), which implies that the two segments are close to the edges $(u_s,u_t)$ and $(u_t,x)$ in $\Gamma_s$, respectively. This guarantees that any edge has exactly one bend, as any edge is outgoing for exactly one of its endvertices.

We next discuss how we place the bends for the outgoing edges of $u$. Since each split-edge belongs to $M_2$, it is drawn in $\Gamma_s$ either as the diagonal of a $1 \times 1$ grid box or as a closing edge. Also, since the outgoing edges of $u$ are in $M_1$ and $M_3$, they are drawn as horizontal and vertical line-segments.

Assume first that the split-edge of $u$ is the diagonal of a $1 \times 1$ grid box; see \cref{fig:port-placement-deg4-a}. If the outgoing edge $(u_t,x)$ belongs to $M_1$, then we place the bend of $(u,x)$ either half a unit to the right of $u_t$ if $x$ is to the right of $u_t$ in $\Gamma_s$, or half a unit to its left otherwise. Symmetrically, if  $(u_t,x)$ belongs to $M_3$, we place the bend half a unit either above or below $u_t$. 

Assume now that the split-edge of $u$ is a closing edge in exactly one of $H_y$ or $H_x$, say w.l.o.g.\ of a cycle $c$ in $H_x$, i.e., it spans the whole $x$-interval of $c$. 
By construction, the outgoing edge of $(u_t,x)$ that belongs to $M_3$ is a vertical line-segment attached above $u$, as either $u_t$ or $u_s$ are the first vertex of $c$ in $\prec_y$; in the latter case, $u_t$ is the second vertex of $c$ in $\prec_y$ by construction. 
If the edge $(u_t,x)$ belongs to $M_3$, we place the bend exactly at the computed position of $u_t$. If $(u_t,x)$ belongs to $M_1$, we place it either half a unit to the right of $u_t$ if $x$ is to the right of $u_t$ in $\Gamma_s$, or half a unit to its left otherwise; see \cref{fig:port-placement-deg4-b}.

Assume last that the split-edge of $u$ is the closing edge $e$ for a $H_y$ and a $H_x$ cycle, which is unique by \cref{cor:single-closing}. As discussed for the analogous case in \cref{sec:degree-3} (see the discussion following \cref{cor:single-closing}), one of $u_s$ and $u_t$, say w.l.o.g. $u_s$, is the leftmost, while the other $u_t$ is the bottommost vertex in $\Gamma_s$. For the placement of the bends, we slightly deviate from our approach above. Let $(u_t,x)$ and $(u_t,y)$ be the two edges of $M_1$ and $M_3$ incident to $u_t$ in $G_s$. Then, it is not difficult to find two grid points $b_x$ and $b_y$ sufficiently below the positions of $x$ and $y$ in $\Gamma_s$, such that $(u,x)$ and $(u,y)$ drawn by bending at $b_x$ and $b_y$ do not cross. Since no two bends overlap, no new crossings are introduced and the slopes of the segments involved in crossings are not modified, \mbox{the obtained drawing $\Gamma'$ is a $1$-bend RAC drawing for $G'$ (and thus~for~$G$).}

\begin{figure}[t]
    \centering
    \begin{subfigure}[b]{.48\textwidth}
    \centering
    \includegraphics[scale=0.9,page=1]{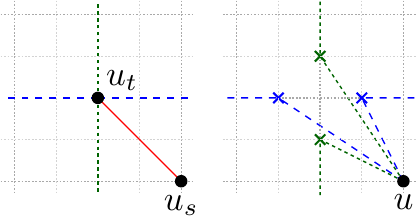}
    \subcaption{}
    \label{fig:port-placement-deg4-a}
    \end{subfigure}
    \hfil
    \begin{subfigure}[b]{.48\textwidth}
    \centering
    \includegraphics[scale=0.9,page=2]{images/bend-point-deg-4.pdf}
    \subcaption{}
    \label{fig:port-placement-deg4-b}
    \end{subfigure}
    \caption{Illustration on how to place the bends in the proof of \cref{thm:degree-4-one}. To merge the vertices $u_s$ and $u_t$ of a vertex $u$ in $G'$, $u$ is placed at the position of $u_s$. The bends of the outgoing edges at $u$ are placed close to the position of $u_t$ in the drawing depending on their orientation.}
    \label{fig:port-placement-deg4}
\end{figure}

Regarding the time complexity, we observe that we can apply 
\cref{thm:2-factors} and the split-operation in $\mathcal{O}(n)$ time. The split operation immediately yields a valid $3$-coloring of the edges, hence we can apply the algorithm of \cref{thm:deg-3-col} to obtain $\Gamma_s$ in $\mathcal{O}(n)$ time. Finally, contracting the edges can clearly be done in $\mathcal{O}(n)$ time, as it requires a constant number of operations per edge. For the area, we observe that in order to place the bends, we have to introduce new grid-points, but we at most double the number of points in any dimension, hence we still maintain the asymptotic quadratic area guaranteed by \cref{thm:deg-3-col}.
\end{proof}

The following theorem, whose proof is in the appendix, provides an alternative construction which additionally guarantees a linear number of edges drawn as straight-line segments.

\begin{theorem2rep}\apxmark\label{thm:degfourold}
Given a degree-$4$ graph $G$ with $n$ vertices and $m$ edges, it is possible to compute in $O(n)$ time a $1$-bend RAC drawing of $G$ with $O(n^2)$ area where at least $\frac{m}{8}$ edges are drawn as straight-line segments.
\end{theorem2rep}
\begin{proof}
The algorithm can be seen as an extension of the one for degree-$3$ graphs introduced in \cite{DBLP:journals/jgaa/AngeliniCDFBKS11}.
Let $F_1$ and $F_2$ be the two edge disjoint $2$-factors of $G$ obtained applying \cref{thm:2-factors}.
We will define a total order $\prec$ of the vertices based on $F_1$ that will guarantee some additional properties.
Once $\prec$ will be computed, the vertices will be placed on the diagonal of a $n \times n$ grid according to $\prec$, i.e., the vertex of position $i$ in $\prec$ will be placed on the coordinate $(i,i)$. Observe that in this way, the only edges of $F_1$ that cannot be drawn on the diagonal are exactly the closing edges of $F_1$.\\
Each edge of $F_2$ will be drawn as a polyline consisting of exactly two segments, a vertical segment followed by a horizontal one, where the orientation of any edge of $F_2$ is given by \cref{thm:2-factors} (recall that every vertex has exactly one incoming and one outgoing edge in $F_2$). The relative position of the endvertices of an edge in $\prec$ will decide whether the edge will be drawn above or below the diagonal. Namely, let $(u,v)$ be a directed edge of $F_2$. If $u \prec v$, we will draw $(u,v)$ below the diagonal. Otherwise, if $ v \prec u$, $(u,v)$ will be drawn above the diagonal. Observe that this scheme never uses the same port twice, as every vertex has indegree and outdegree one in $F_2$. Hence it only remains to show how to add the closing edges of $F_1$.
Remark that if $G$ is Hamiltonian, we can already stop our description at this point by setting $F_1$ to the Hamiltonian cycle, i.e., we can define the order $\prec$ to be consistent with an (arbitrary) Hamiltonian Path (and by observing that the single closing edge of $F_1$ can be added by using two slanted segments).
Let $\Gamma_{\overline{C}}$ be the induced drawing of $G \setminus C$, where $C$ is the set of closing edges.

\begin{lemma}\label{lem:oppo}
Any vertex $v$ in $\Gamma_{\overline{C}}$ has two free ports. Moreover, the free ports are opposite to the used ones.
\end{lemma}
\begin{appproof}
Suppose first that $v$ uses the $W$- and $E$-port. A used $W$-port implies the existence of an edge $(v,u_i)$ (with $v \prec u_i$), while a used $E$-port implies the existence of an edge $(v,u_j)$ (with $u_j \prec v$), but then $v$ has two outgoing edges in $F_2$, a contradiction.  Equivalently, a used $N$-port and a used $S$-port imply incoming edges, hence $v$ has indegree $2$, a contradiction.
\end{appproof}
We now define $\prec$ such that it satisfies our required properties.
\begin{lemma} \label{lem:free-port}
There exists a total order $\prec$ on the vertices such that one of the two extremal vertices with respect to $\prec$ of any cycle of $F_1$ has a good free port.
\end{lemma}
\begin{appproof}
Suppose for a contradiction that no such ordering exists. Clearly, a cyclic rotation w.r.t. $\prec$ of the vertices inside the same cycle maintains the properties of our ordering, hence it is sufficient to describe this rotation for the cycles independently. Let $c$ be a cycle that does not satisfy this property. First observe that not all vertices of $c$ can be adjacent to only chords, as otherwise $G$ would be disconnected. Hence, let $v$ be a vertex that has an external edge. If this external edge is a backward edge, then rotating $c$ such that $v$ is the first vertex of $c$ in $\prec$ implies that either the $N$-port or the $E$-port of $v$ is free.
Similarly, if it is a forward edge, rotating $c$ such that $v$ is the last vertex of $c$ in $\prec$ implies a free $S$-port or a free $W$-port.
\end{appproof}

Now we describe how to add  the closing edges $C$ to $\Gamma_{\overline{C}}$.
By \cref{lem:free-port}, we have at least one good port free for any closing edge. Let $c$ be a cycle of $F_1$. W.l.o.g. assume that the last vertex of $c$ in $\prec$ has a free $S$-port. The other cases follow symmetrically.
Let $v_1,v_2,\dots,v_n$ be the ordering of the vertices of $c$ induced by $\prec$. Assume that  $(i,i)$ is the current position of $v_i$ (clearly, the actual position is offset by the number of vertices that are before $v_1$ in $\prec$, but this is omitted for clarity reasons).

\begin{enumerate}

\item $v_1$ has a free $E$-port \\
Then we can simply add the closing edge to $\Gamma_{\overline{C}}$.
Hence, we now assume that $v_1$ has no free $E$-port. By \cref{lem:oppo} it follows that the $W$-port of $v_1$ is free.
\item The $E$-port and the $S$-port of $v$ are taken. Observe that the edge using the $S$-port of $v_1$ is a backward edge.
 If the edge that uses the $E$-port of $v_1$ is external (forward) \\
In this case, we move $v_1$ to $(n,2-\epsilon)$. This allows us to add the edge $(v_1,v_n)$ using the $N$-port of $v_1$, which is free by assumption. Note that, as the edge using the $E$-port of of $v_1$ is a future edge, we can keep the $E$-port for it. Similarly, this holds for the edge using the $S$-port, which is a backward edge by definition. Now, the edge $(v_1,v_2)$ will be drawn as a polyline with a horizontal segment at $v_1$ which uses the $E$-port and with a bend that is sufficiently close to $v_2$ such that we use a non orthogonal port of $v_2$; see \cref{fig:move-diag-1}.
Otherwise, the edge that uses the $E$-port of $v_1$ is a chord. Let $v_i \in c$ be the other endpoint of the chord.
We place $v_1$ at $(i,2-\epsilon)$. Edge $(v_1,v_n)$ uses the $E$-port of $v_1$, we draw $(v_1,v_i)$ such that it uses the $N$-port of $v_1$ and the $S$-port of $v_i$ (which is free as it was used by the same edge previously) and redraw the edge $(v_1,v_2)$ as described in the previous case as shown in \cref{fig:move-diag-2}.

\item The $E$-port and the $N$-port of $v_1$ are not free
Suppose first that both edges using the $E$-port and $N$-port are external (forward),
We move $v_1$ to $(n+\epsilon,2-\epsilon)$. This way, we can use a non orthogonal port for  at $v_1$. Observe that by assumption the $W$-port of $v_1$ is also free, hence $(v_1,v_n)$ could be drawn such that it uses the $W$-port of $v_1$. Both forward edges of $v_1$ will retain the same port; see \cref{fig:move-diag-3}.

Suppose now that the edge using the $E$-port is external, while the edge using the $N$-port is internal  \\
Let $v_i$ be the endpoint of the chord that is incident to  $v_1$. By Lemma \ref{lem:oppo}, we know that the $E$-port of $v_i$ is free. Hence, we can place $v_1$ at  $(n-\epsilon,2-\epsilon)$, which allows to add $(v_1,v_i)$ using the $N$-port of $v_1$ and the $E$-port of $v_i$. Further, placing $v_1$ at the $x$-position $(n- \epsilon)$ allows us to add $(v_1,v_n)$ using a non orthogonal port at $v_1$, while the forward edge of $v_1$ will keep using the $E$-port of $v_1$; see \cref{fig:move-diag-3}. Suppose last that the edge using the $E$-port is internal, while the edge using the $N$-port is external. Let $v_i$ be the endpoint of the chord that is incident to $v_1$.
By Lemma \ref{lem:oppo}, either the $W$-port or the $E$-port of $v_i$ is free, w.l.o.g. assume that the $E$-port of $v_i$ is free as the other case is symmetric. Then we place $v_1$ to $(i-\epsilon,2-\epsilon)$. Further, we redraw the edge $(v_{i-1},v_i)$ such that it is not on the diagonal, but rather it has a horizontal segment starting at $v_i$ and then a slanted segment to $v_{i-1}$. This allows the future edge of $v_1$ to cross the spine; see \cref{fig:move-diag-5,fig:move-diag-6}.
\item Both edges of $F_2$ that are incident to $v_1$ are chords. \\
Let $v_t$ be the endpoint of the edge that is using the $N$-port of $v_1$ and $v_r$ be the endpoint of the edge that is using the $E$-port of $v_1$. Observe that the $S$-port of $v_r$ is free (as it was occupied by the edge $(v_1,v_r)$ before) as well as the $E$-port of $v_t$ by \cref{lem:oppo}. 
First assume that $t < r$.
We place $v_1$ at $(r -\epsilon,2-\epsilon)$  which allows to add $(v_1,v_r)$ using a non orthogonal port of $v_1$ and then we use the $N$-port of $v_1$ to connect it to $v_t$ as shown in \cref{fig:move-diag-7} \\
Now assume that $r < t$. By Lemma \ref{lem:oppo}, either the $W$-port or the $E$-port of $v_r$ is free, w.l.o.g. assume that the $E$-port is free.
Then, we place $v_1$ at $(r +\epsilon,2-\epsilon)$. Again, we redraw the edge $(v_r,v_{r+1})$ such that the edge $(v_1,v_t)$ can cross the spine; see \cref{fig:move-diag-8,fig:move-diag-9}. The other edges are analogous to the previous cases.
\end{enumerate}

\begin{figure}[p]
    \centering
    \begin{subfigure}[b]{.31\textwidth}
    \centering
    \includegraphics[scale=1,page=1]{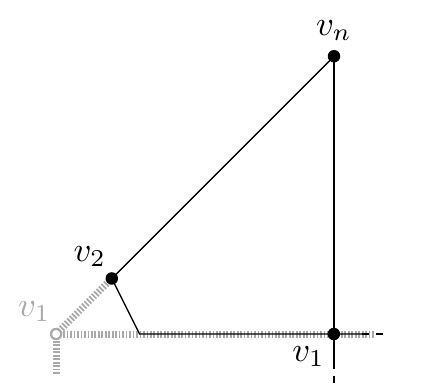}
    \subcaption{}
    \label{fig:move-diag-1}
    \end{subfigure}
    \hfil
    \begin{subfigure}[b]{.31\textwidth}
    \centering
    \includegraphics[scale=1,page=2]{images/move-diag-new.pdf}
    \subcaption{}
    \label{fig:move-diag-2}
    \end{subfigure}
    \begin{subfigure}[b]{.31\textwidth}
    \centering
    \includegraphics[scale=1,page=3]{images/move-diag-new.pdf}
    \subcaption{}
    \label{fig:move-diag-3}
    \end{subfigure}\\
    \bigskip
    \begin{subfigure}[b]{.31\textwidth}
    \centering
    \includegraphics[scale=1,page=4]{images/move-diag-new.pdf}
    \subcaption{}
    \label{fig:move-diag-4}
    \end{subfigure}
    \hfil
    \begin{subfigure}[b]{.31\textwidth}
    \centering
    \includegraphics[scale=1,page=5]{images/move-diag-new.pdf}
    \subcaption{}
    \label{fig:move-diag-5}
    \end{subfigure}
    \begin{subfigure}[b]{.31\textwidth}
    \centering
    \includegraphics[scale=1,page=6]{images/move-diag-new.pdf}
    \subcaption{}
    \label{fig:move-diag-6}
    \end{subfigure}\\
    \bigskip
    \begin{subfigure}[b]{.31\textwidth}
    \centering
    \includegraphics[scale=1,page=7]{images/move-diag-new.pdf}
    \subcaption{}
    \label{fig:move-diag-7}
    \end{subfigure}
    \hfil
    \begin{subfigure}[b]{.31\textwidth}
    \centering
    \includegraphics[scale=1,page=8]{images/move-diag-new.pdf}
    \subcaption{}
    \label{fig:move-diag-8}
    \end{subfigure}
    \begin{subfigure}[b]{.31\textwidth}
    \centering
    \includegraphics[scale=1,page=9]{images/move-diag-new.pdf}
    \subcaption{}
    \label{fig:move-diag-9}
    \end{subfigure}
    \caption{Illustration of the cases in the proof of \cref{thm:degfourold} to add the closing edges in $C$ to $\Gamma_{\overline{C}}$. The grey dotted lines show the initial drawing of edges which are redrawn in order to add $C$.}
    \label{fig:move-diag}
\end{figure}

In order to bound the number of edges that require a bend, let $c = (v_1,v_2,\dots,v_n)$ be a cycle of $F_1$. By construction, all edges of $F_1 $ are drawn straight line with the exception of possibly (i) the first edge of $c$, (ii) the last (i.e., the closing edge) of $c$ and (iii) one intermediate edge.

Among the two $2$-factors, we choose $F_1$ as the one that contains more real edges. Let $c$ be a cycle of $F_1$. Denote by $r(c)$ the number of real edges of $c$ and $s(c)$ the number of edges of $c$ that are drawn without a bend.\\
Clearly, if $c$ contains a fake edge, we can rotate $c$ such that the closing edge corresponds to the fake edge, in which case $s(c) = r(c)$. 
Now let $c$ be a cycle that contains no fake edge. Clearly, we have $|c| \geq 3$, as we assume the initial graph to be simple. Assume that  $|c| = 3$. Clearly, there can be no chords in $c$, hence all edges of $F_2$ incident to the vertices of $c$ are external. Now, as any pair of vertices in $c$ is adjacent, if one vertex $v_1$ is incident to two forward edges (backward edges), while the other vertex $v_2$ is incident to at most one forward (backward) edge, we can simply define the closing edge to be $(v_1,v_2)$ and set $v_1$ as the smallest (highest) and $v_2$ as the highest (smallest), which allows us to add the closing edge using a bend but without moving a vertex.
Hence, assume that all vertices have two forward edges. But then in particular, the $S$-port of the middle vertex is free, hence we can draw the closing edge of $c$ using a slanted port at the lowest vertex without introducing a crossing. Hence, we conclude that in both cases, $s(c) = r(c)-1$. 
For $|c| \geq 4$, at most $|c|-3$ edges will get a bend by the previous observation.

By the choice of $F_1$, we have that $r(F_1) \geq \frac{m}{2}$. Clearly, $r(F_1) = \sum_i r(c_i)$, where the union of the cycles $c_i$ is $F_1$. If $c$ is a cycle with a fake edge, we have that $s(c) = r(c)$, i.e., all real edges remain on the diagonal. Otherwise, we have that $s(c) \geq \frac{r(c)}{4}$, hence it follows that
$$s(F_1) \geq \frac{r(F_1)}{4} \geq \frac{m}{8}$$
edges can be drawn without a bend.

We conclude the proof of this theorem by discussing the time complexity and the required area. Applying \cref{thm:2-factors} to obtain the $2$-factors can be done in $\mathcal{O}(n)$ time, which gives us an initial total ordering of the vertices $\prec_{init}$. Defining a feasible total order of the vertices for every cycle $\prec$ given $\prec_{init}$ can be done in $\mathcal{O}(n)$ time. The potential displacement of the first vertex of a cycle in $\prec_y$ requires a constant number of operations, hence in total $\mathcal{O}(n)$ time.
For the area, we observe that the initial positioning on the diagonal takes $n \times n$ area (including the bends). By setting $\epsilon$ to $\frac{1}{3}$, it is sufficient to scale the final drawing by a factor of $3$ in order to guarantee grid coordinates for vertices and bends.
\end{proof}


\section{RAC drawings of $\mathbf{7}$-edge-colorable degree-$\mathbf{7}$ graphs}
We prove that $7$-edge-colorable degree-$7$ graphs admit $2$-bend RAC drawings by proving the following slightly stronger statement.

\begin{theorem}\label{thm:7colorable}
Given a degree-$7$ graph $G$ decomposed into a degree-$6$ graph $H$ and a matching $M$, it is possible to compute in $O(n)$ time a $2$-bend RAC drawing of $G$ with $O(n^2)$ area. 
\end{theorem}

Since $H$ is a degree-$6$ graph, it admits a decomposition into three disjoint (directed) $2$-factors $F_1$, $F_2$ and $F_3$  after applying \cref{thm:2-factors} and (possibly) augmenting $H$ to a $6$-regular (multi)-graph. To distinguish between directed and undirected edges, we write $\{u,v\}$ to denote an undirected edge between $u$ and $v$, while $(u,v)$ denotes a directed edge from $u$ to $v$.
In the following, we will define two total orders $\prec_x$ and $\prec_y$, which will define the $x$- and $y$-coordinates of the vertices of $G$, respectively. We define $\prec_y$ such that the vertices of each cycle in $F_1$ will be consecutive in $\prec_y$. Initially, for any cycle of $F_1$, the specific internal order of its vertices in $\prec_y$ is specified by one of the two traversals of it; however, we note here that this choice may be refined later in order to guarantee an additional property (described in \cref{lem:no-3-0}).
The definition of $\prec_x$ is more involved and will also be discussed later. \cref{thm:2-factors} guarantees that the edges of $F_2$ ($F_3$) are oriented such that any vertex has at most one incoming and one outgoing edge in $F_2$ ($F_3$). 
Once $\prec_y$ and $\prec_x$ are computed, each vertex $u$ of $G$ will be mapped to point $(8i, 8j)$ of the Euclidean plane provided that $u$ is the $i$-th vertex in $\prec_x$ and the $j$-th vertex in $\prec_y$.
Each vertex $u$ is associated with a closed box $B(u)$ centered at $u$ of size $8 \times 8$. 
We aim at computing a drawing of $G$ in which 
\begin{enumerate*}[label={(\roman*)}, ref=(\roman*)]
\item\label{prp:no-box-overlap}no two boxes overlap, and 
\item\label{prp:in-the-box}the edges are drawn with two bends each so that only the edge-segments that are incident to $u$ are contained in the interior of $B(u)$, while all the other edge-segments are either vertical or horizontal. 
\end{enumerate*}
\mbox{This guarantees that the resulting drawing is $2$-bend RAC; see \cref{fig:k8}.}

\begin{figure}[t]
    \centering
    \begin{subfigure}[b]{.48\textwidth}
    \centering
    \includegraphics[scale=0.5,page=1]{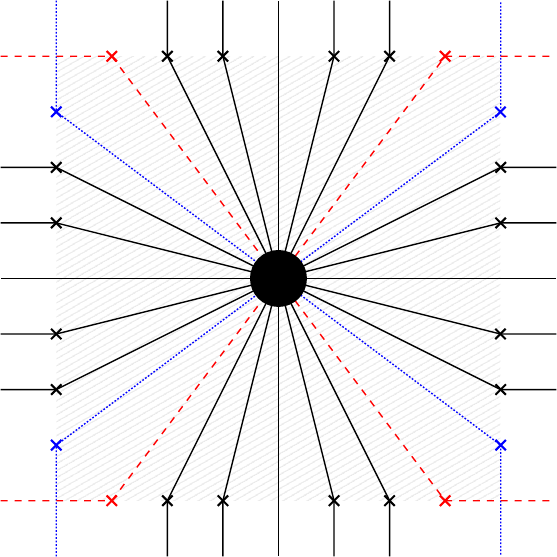}
    \subcaption{}
    \label{fig:the-box}
    \end{subfigure}
    \hfil
    \begin{subfigure}[b]{.48\textwidth}
    \centering
    \includegraphics[scale=0.5,page=2]{images/the-box-example.pdf}
    \subcaption{}
    \label{fig:k8}
    \end{subfigure}
    \caption{Edge routing in the $8 \times 8$ box $B(u)$ of a vertex $u$ (a). Red dashed (blue dotted) ports are reserved for horizontal (vertical) type-$2$ edges. A $2$-bend RAC drawing of $K_8$ is shown in (b). In this drawing, red dashed edges are horizontal type-$2$, while the blue dotted one is vertical type-$2$.}
\end{figure}
In the final drawing, all edges will be drawn with exactly three segments, out of which either one or two are \emph{oblique}, i.e., they are neither horizontal nor vertical.  It follows from \ref{prp:in-the-box} that the bend point between an oblique segment and a vertical (horizontal) segment lies on a horizontal (vertical) side of the box containing the oblique segment.
During the algorithm, we will classify the edges as \emph{type-$1$} or \emph{type-$2$}. Type-$1$ edges will be drawn with one oblique segment, while type-$2$ edges with two oblique segments. In particular, for a type-$1$ edge $(u,v)$, we further have that the oblique segment is incident to $v$, which implies that $(u,v)$ occupies an orthogonal port at $u$. On the other hand, a type-$2$ edge $(u,v)$ requires that $B(u)$ and $B(v)$ are \emph{aligned} in $y$ (in $x$), i.e., there exists a horizontal (vertical) line that is partially contained in both $B(u)$ and $B(v)$, in order to draw the middle segment of $(u,v)$ horizontally (vertically). By construction, this is equivalent to having $u$ and $v$ consecutive in $\prec_y$ ($\prec_x$). These alignments guarantee that if we partition the edges of $F_1$ into $\bar{F_1}$ and $\hat{F_1}$ containing the closing and non-closing ones, respectively, then it is possible to draw $\hat{F_1}$ as a horizontal type-$2$ edge (independent of the $x$-coordinate of its endvertices), as its endvertices are consecutive in $\prec_y$ by construction. 
Thus, we can put our focus on edges in $\bar{F_1} \cup F_2 \cup F_3$ , which we initially classify as type-$1$ edges (by orienting each edge $(u,v)$ of $\bar{F_1}$ from $u$ to $v$ if $u \prec_y v$).
We refine $\prec_y$ using the concept of \emph{critical vertices}. Namely, for a vertex $u$ of $G$, the direct successors of $u$ in $\bar{F_1} \cup F_2 \cup F_3$ are the critical vertices of $u$, which are denoted by $c(u)$.
Based on the relative position of $u$ to its critical vertices in $\prec_y$, we label $u$ as $(\alpha,\beta)$, if $\alpha$ vertices $t_1,\ldots,t_\alpha$ of $c(u)$ are after $u$ in $\prec_y$ and $\beta$ vertices $b_1,\ldots,b_\beta$ before. We refer to $t_1,\ldots,t_\alpha$ ($b_1,\ldots,b_\beta$) as the \emph{upper} (\emph{lower}) \emph{critical neighbors} of $u$. An edge connecting $u$ to an upper (lower) critical neighbor is called \emph{upper critical} (\emph{lower critical}, respectively). More general, the upper and lower critical edges of $u$ are its \emph{critical edges}. 

Note that $2 \leq \alpha+\beta \leq 3$ as any vertex has exactly one outgoing edge in $F_2$ and $F_3$ and at most one in $\bar{F_1}$, that is, the number of upper and lower critical neighbors of vertex $u$ ranges between $2$ and $3$. It follows that the label of each vertex of $H$ is in $\{(0,2), (1,1), (2,0), (1,2), (2,1), (3,0)\}$; refer to these labels as the \emph{feasible labels} of the vertex. Observe that a $(3,0)$-, $(2,1)$- or $(1,2)$-label implies that the vertex is incident to a closing edge of $F_1$ (hence, each cycle in $F_1$ has at most one such vertex, which is its first one in $\prec_y$). This step will complete the definition of $\prec_y$.

\begin{lemma2rep}\apxmark\label{lem:no-3-0}
For each cycle $c$ of $F_1$, there is an internal ordering of its vertices followed by a possible reorientation of one edge in $F_2 \cup F_3$, such that in the resulting $\prec_y$ 
\begin{enumerate*}[label={(\alph*)}, ref=(\alph*)]
\item\label{prp:feasible}every vertex of $c$ has a feasible label,
\item\label{prp:no-3-0}no vertex of $c$ has label $(3,0)$, and 
\item\label{prp:1-2-in}if there exists a $(1,2)$-labeled vertex in $c$, then its (only) upper critical neighbor belongs to $c$.
\end{enumerate*}
\end{lemma2rep}

\begin{proof}
Let $c = (v_1,\dots,v_k)$ be a cycle of $F_1$ with the initial internal ordering $v_1 \prec_y \ldots \prec_y v_k$. First observe that the labels of any vertex is dependent on $\prec_y$ and therefore a change to the internal ordering of $c$ can potentially change the labels of the vertices of $c$. Also, the closing edge is dependent on this internal ordering - we will always assume that the current closing edge $(u,v)$ of $c$ is directed from $u$ to $v$ if $u \prec_y$ v.
An internal ordering of $c$ is completely defined by specifying the first and the last vertex of $c$ in $\prec_y$ (which are connected in $F_1$), i.e., there is an $i$ such that $v_i$ and $v_{i-1}$ are the first and the last vertex of $c$ (or vice-versa).
We first consider a case that is simple to be addressed. In particular, if there exists a vertex $v_i$ of $c$ that is incident to multiple copies of the same edge and at least one of those copies is a critical edge for $v_i$, then by removing this copy we can ensure that $v_i$ has at most one critical edge besides the potential closing edge of $c$. We place $v_i$ as the first vertex of $c$ in $\prec_y$ and $v_{i-1}$ as the last.  Since $v_i$ is now incident to the closing edge of $c$, it has at most two upper critical neighbors, guaranteeing that $v_i$ does not have the label $(3,0)$ or $(1,2)$, as $v_i$ has at most two critical neighbors and no other vertex in $c$ has the label $(1,2)$, as they do not have an outgoing edge in $\bar{F_1}$ and since we did not redirect any edge in $F_2$ or $F_3$, it follows that they have at most one outgoing edge in $F_2$ and one in $F_3$, hence Properties~\ref{prp:no-3-0} and \ref{prp:1-2-in} hold. Also, since we did not reorient any edge in $c$, Property~\ref{prp:feasible} is maintained. Otherwise, we proceed by distinguishing two main cases. Assume first that there is a vertex $v_i$ of $c$ having one or two backward edges to a cycle $c'$ such that the vertices of $c'$ precede the vertices of $c$ in $\prec_y$ (note that these backward edges are by definition in $F_2 \cup F_3$). We place $v_i$ as the first vertex of $c$ in $\prec_y$ and $v_{i-1}$ as the last. Since $v_i$ is incident to a backward edge, it has at least one lower critical neighbor, hence $v_i$ does not have the label $(3,0)$ which guarantees Property~\ref{prp:no-3-0}). The only case in which $v_i$ can have label $(1,2)$ is if it has two lower critical neighbors. But in this case, its upper critical neighbor corresponds to  $v_{i-1}$ and is in the same cycle, hence Property~\ref{prp:1-2-in} is maintained. Also, since no edge was reoriented, it follows that the label of each vertex of $c$ remains feasible as required by Property~\ref{prp:feasible}. 

\begin{figure}[t]
    \centering
    \begin{subfigure}[b]{.31\textwidth}
    \centering
    \includegraphics[scale=1,page=1]{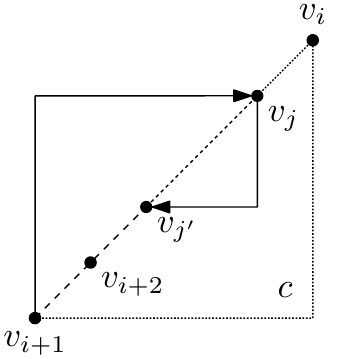}
    \subcaption{}
    \label{fig:cycle-flip-1}
    \end{subfigure}
    \hfil
    \begin{subfigure}[b]{.31\textwidth}
    \centering
    \includegraphics[scale=1,page=2]{images/cycle-flip.pdf}
    \subcaption{}
    \label{fig:cycle-flip-2}
    \end{subfigure}
    \begin{subfigure}[b]{.31\textwidth}
    \centering
    \includegraphics[scale=1,page=3]{images/cycle-flip.pdf}
    \subcaption{}
    \label{fig:cycle-flip-3}
    \end{subfigure}
    \caption{Cycle $c$ with (a) $v_{i+1} \prec_y v_{j'} \prec_y v_j$ and (b) a (possibly reordered) cycle $c$ with $v_{i+1} \prec_y v_j \prec_y v_{j'}$. Subfigure (c) shows $c$ with the redirected edge $(v_j, v_{i+1})$. }
\end{figure}

Suppose now that no vertex of $c$ has a backward edge, i.e., for each vertex in $c$, its outgoing edges are either chords in $c$ or forward edges. Among the vertices of $c$, let $v_i$ be the one that is incident to the most forward edges. We place $v_i$ as the last vertex of $c$ in $\prec_y$ and $v_{i+1}$ as the first. If the number of forward edges at $v_i$ is two, then it has label $(2,0)$.
In this case, we reorient the closing edge $(v_{i+1},v_i)$ of $c$ from $v_i$ to $v_{i+1}$, which implies that $v_i$ now has label $(2,1)$ and $v_{i-1}$ has label $(2,0)$; since any other vertex has at most two critical neighbors, it follows that Properties~\ref{prp:feasible},~\ref{prp:no-3-0} and~\ref{prp:1-2-in} hold.
Otherwise, $v_i$ has either label $(1,1)$ or $(0,2)$. By definition of $v_i$, any other vertex in $c$ is incident to at most one forward edge. Observe that $v_{i+1}$ has label $(3,0)$, as otherwise we would be in the second case. Since each vertex in $c$ has at least one outgoing edge to another vertex in $c$ (since we have no backward edges and at most one forward edge), there exists a path $v_{i+1} \rightarrow v_j \rightarrow v_{j'}$  with $1 \leq j,j' \leq k$ in $F_2 \cup F_3$. Note that $v_j$ cannot be the last vertex of $c$, since this would imply an outgoing multi-edge; a case which has been addressed at the beginning of the proof. Further, we have that $v_j'$ does not coincide with $v_{i+1}$, as otherwise there is an outgoing multi-edge at $v_j$, a case we covered before. If $v_{i+1} \prec_y v_j \prec_y v_{j'}$, then we reorient the edge  $(v_{i+1}, v_j)$ from $v_j$ to $v_{i+1}$; see \cref{fig:cycle-flip-2,fig:cycle-flip-3}. While $v_{i+1}$ has now label $(2,0)$ (since it has at most two critical neighbors), $v_j$ can have label $(2,1)$ or $(1,2)$, hence  Property~\ref{prp:no-3-0} holds. If $v_j$ has label $(2,1)$, Property~\ref{prp:1-2-in} clearly holds. Otherwise, $v_j$ the label $(1,2)$, but  since $v_j \prec_y v_{j'}$ and $v_{j'} \in c$, Property~\ref{prp:1-2-in} still holds.
Since the vertices $v_{i+1}$ and $v_j$ that are incident to the reoriented edge $(v_{i+1}, v_j)$ have feasible labels and since the labels of the remaining vertices of $c$ are not affected, Property~\ref{prp:feasible} is maintained.

To complete the proof, it remains to consider the case in which $v_{i+1} \prec_y v_{j'} \prec_y v_j$. We place $v_{i+1}$ as the first vertex of $c$ in $\prec_y$ and $v_{i+2}$ as the last. This internal reordering guarantees that $v_{i+1} \prec_y j \prec_y j' \prec_y v_{i+2}$ and we proceed as in the previous case. This operation is illustrated in \cref{fig:cycle-flip-1,fig:cycle-flip-2}.\qedhere
\newpage
\end{proof}

Now that $\prec_y$ is completely defined, we orient any edge $(u,v) \in M$ from $u$ to $v$ if and only if $u \prec_y v$. In this case, we further add $v$ as a critical vertex of $u$. This implies that some vertices can have one more critical upper neighbor, which then gives rise to the new following labels, which we call \emph{tags} for distinguishing: $\{[3,1],[3,0],[2,2],[2,1],[2,0],[1,2],[1,1],[0,2]\}$. In this context, \cref{lem:no-3-0} guarantees the following property. 

\begin{prop}\label{prop:one-three}
Any cycle $c$ of $F_1$ has at most one vertex with tag $[\alpha,\beta]$ such that $\alpha+\beta=4$.
\end{prop}

Next, we compute the final drawing satisfying Properties \ref{prp:no-box-overlap} and \ref{prp:in-the-box} by performing two iterations over the vertices of $G$ in reverse $\prec_y$ order. In the first, we specify the final position of each vertex of $G$ in $\prec_x$ and classify its incident edges while maintaining the following \cref{inv:consecutive}. In the second one, we exploit the computed $\prec_x$ to draw all edges of~$G$.
 
\begin{invariant}\label{inv:consecutive}
The endvertices of each vertical type-$2$ edge are consecutive in~$\prec_x$. Further, any vertex is incident to at most one vertical type-$2$ edge.
\end{invariant} 

The second part of \cref{inv:consecutive} implies that the vertical type-$2$ edges form a set of independent edges. In this regard, we say that a vertex $u$ is a \emph{partner} of a vertex $v$ in $G$ if and only if $u$ and $v$ are connected with an edge in this set. 

In the first iteration, we assume that we have processed the first $i$ vertices $v_n,\ldots,v_{n-i+1}$ of $G$ in reverse $\prec_y$ order and we have added these vertices to $\prec_x$ together with a classification of their incident edges satisfying \cref{inv:consecutive}.
We determine the position of $v_{n-i}$ in $\prec_x$ based on the $\prec_x$ position of its upper critical neighbors. The incident edges of $v_{n-i}$ are classified based on a case analysis on its tag $[\alpha,\beta]$. Recall that unless otherwise specified, every edge is a type-$1$ edge.

\begin{enumerate}
\item The tag of $v_{n-i}$ is $[3,1]$ or $[3,0]$: 
Let $a$, $b$ and $c$ be the upper critical neighbors of $v_{n-i}$, which implies that they were processed before $v_{n-i}$ by the algorithm and are already part of $\prec_x$. W.l.o.g.\ assume that $a \prec_x b \prec_x c$. By \cref{inv:consecutive}, vertex $b$ is the partner of at most one already processed vertex $b'$, which is consecutive with $b$ in $\prec_x$. If $b'$ exists and $b' \prec_x b$, then we add $v_{n-i}$ immediately after $b$ in $\prec_x$. Symmetrically, if $b'$ exists and $b \prec_x b'$, then we add $v_{n-i}$ immediately before $b$ in $\prec_x$. Otherwise, we add $v_{n-i}$ immediately before $b$ in $\prec_x$. This guarantees that $v_{n-i}$ is placed between $a$ and $c$ in $\prec_x$ and that \cref{inv:consecutive} is satisfied, since none of the upper critical edges incident to $v_{n-i}$ was classified as a~type-$2$~edge.

\item The tag of $v_{n-i}$ is $[2,1]$ ,$[2,0]$, $[1,2]$, $[1,1]$ or $[0,2]$: 
By appending \mbox{$v_{n-i}$ to $\prec_x$,} we maintain \cref{inv:consecutive}, since none of the upper critical edges incident to $v_{n-i}$ was classified as type-$2$.

\item The tag of $v_{n-i}$ is $[2,2]$: Let $a$ and $b$ be the upper critical neighbors of $v_{n-i}$, which implies that they were processed before $v_{n-i}$ by the algorithm and are already part of $\prec_x$. W.l.o.g. assume that $(v_{n-i},a) \in M$. We classify the edge $(v_{n-i},b)$ as a vertical type-$2$ edge and we add $v_{n-i}$ immediately before $b$ in $\prec_x$. To show that \cref{inv:consecutive} is maintained by this operation it is sufficient to show that $b$ was not incident to a vertical type-$2$ edge before. Suppose for a contradiction that there is a vertex $b'$ in $\{v_n,\ldots,v_{n-i+1}\}$, such that $(b,b')$ or $(b',b)$ is a type-$2$ edge. As seen in the previous cases, this implies that $b$ or $b'$ has tag $[2,2]$, respectively. Since in the $[2,2]$ case the edge classified as type-$2$ is the one not in $M$ and since any vertex that has tag $[2,2]$ has label $(1,2)$, by \cref{lem:no-3-0} it follows that vertical type-$2$ edges are chords of a cycle. Hence, $b$ or $b'$ would lie in the same cycle as $v_{n-i}$, which is a contradiction to \cref{prop:one-three}, thus \cref{inv:consecutive} holds.
\end{enumerate}

Orders $\prec_x$ and $\prec_y$ define the placement of the vertices. By iterating over the vertices, we describe how to draw the~edges to complete the drawing such that Properties \ref{prp:no-box-overlap} and \ref{prp:in-the-box} are satisfied. We distinguish cases based on the tag of the current vertex $v_i$.

\begin{enumerate}
    \item The tag of $v_i$ is $[3,1]$ or $[3,0]$: 
    Let $\{a,b,c\}$ be the upper critical neighbors of $v_i$. The construction of $\prec_x$ ensures that not all of $\{a,b,c\}$ precede or follow $v_i$ in $\prec_x$, w.l.o.g. we can assume that $a \prec_x b,v_i \prec_x c$. Then, we assign the $W$-port at $v_i$ to $(v_i,a)$, the $N$-port at $v_i$ to $(v_i,b)$ and the $E$-port at $v_i$ to $(v_i,c)$. If $v_i$ has a lower critical neighbor, we assign the $S$-port at $v_i$ for the edge connecting $v_i$ to it.
    
    \item The tag of $v_i$ is $[2,1]$ or $[2,0]$:
    Let $\{a,b\}$ be the upper critical neighbors of $v_i$. We assign the $N$-port at $v_i$ to $(v_i, a)$. Note that $v_i$ was appended to $\prec_x$ during its construction. If $b \prec_x v_i$, we assign the $W$-port at $v_i$ to $(v_i,b)$. Otherwise, we assign the $E$-port at $v_i$ to $(v_i,b)$. The $S$-port is assigned to the lower critical edge of $v_i$, if present.
    
    \item The tag of $v_i$ is $[1,2]$ or $[0,2]$: This case is symmetric to the one above by exchanging the roles of upper and lower critical neighbors and $N$- and $S$-ports.

    \item The tag of $v_i$ is $[1,1]$:
    Let $a$ be the upper critical neighbor and $b$ the lower critical neighbor of $v_i$. Then we assign the $N$-port to the edge $(v_i, a)$ and the $S$-port to $(v_i, b)$.

     \item The tag of $v_i$ is $[2,2]$:
     Let $\{a,b\}$ and $\{c,d\}$ be the upper and lower critical neighbors~of~$v_i$. W.l.o.g. let $(v_i,a) \in M$. By \cref{inv:consecutive} and construction, the edge $(v_i,b)$ is a type-$2$ edge. The $N$- and $S$-ports at $v_i$ are assigned to the edges $(v_i, a)$ and $(v_i, c)$. If $d \prec_x v_i$, we assign the $W$-port at $v_i$ to $(v_i,d)$. Otherwise, we assign the $E$-port at $v_i$ to $(v_i,d)$.
\end{enumerate}

We describe how to place the bends of the edges on each side of the box $B(u)$ of an arbitrary vertex $u$ based on the type of the edge that is incident to $u$, refer to \cref{fig:the-box}. We focus on the bottom side of $B(u)$. Let $(x_u,y_u)$ be the position of $u$ that is defined by $\prec_x$ and $\prec_y$.
Recall that the box $B(u)$ has size $8 \times 8$. Let $e = \{u,v\}$ be an edge incident to $u$. If $e$ is a horizontal type-$2$ edge, then we place its bend at $(x_u-3,y_u-4)$, if $v \prec_x u$, otherwise we have $u \prec_x v$ and we place the bend at  $(x_u+3,y_u-4)$. If $e$ is a type-$1$ edge that uses the $S$-port of $u$, then segment of $e$ incident to $u$ passes through point $(x_u,y_u-4)$.
If $e$ is a type-$1$ edge oriented from $v$ to $u$ such that $v \prec_y u$ and $e$ uses either the $W$-port or the $E$-port of $v$, then we place the bend at $(x_u+i,y_u-4)$ with $i \in \{-2,-1,1,2\}$. Since any vertex has at most four incoming type-$1$ edges after applying \cref{lem:no-3-0}, we can place the bends so that no two overlap. No other edge crosses the bottom side of $B(u)$. The description for the other sides can be obtained by rotating this scheme; for the left and the right side the type-$2$ edges are the vertical ones.

We now describe how to draw each edge $e=(u,v)$ of $G$ based on the relative position of $u$ and $v$ in $\prec_x$ and $\prec_y$ and the type of $e$. Refer to \cref{fig:k8}.
Suppose first that $e$ is a type-$2$ edge. If $e$ is a horizontal type-$2$ edge, then $u$ and $v$ are consecutive in $\prec_y$ and $B(u)$ and $B(v)$ are aligned in $y$-coordinate, in particular, there is a horizontal line that contains the top side of one box and the bottom side of the other, hence it passes through the two assigned bend-points, which implies that the middle segment is horizontal. Similarly, if $e$ is a vertical type-$2$ edge, then $u$ and $v$ are consecutive in $\prec_x$ by \cref{inv:consecutive}. Hence, the assigned points for the bends define a vertical middle segment.
Suppose now that $e$ is a type-$1$ edge. 
The case analysis for the second iteration over the vertices guarantees that for any relative position of $v$ to $u$, we assigned an appropriate orthogonal port at $u$ which allows to find a point on the first segment, such that the orthogonal middle-segment of the edge $e$ (that is perpendicular to the first) can reach the assigned bend point on the boundary of $B(v)$.

We argue that the constructed drawing is indeed $2$-bend RAC as follows. By construction, every edge consists of three segments and no bend overlaps with an edge or with another bend. Each vertical (horizontal) line either crosses only one box or contains the side of exactly two boxes, whose corresponding vertices are consecutive in $\prec_x$ ($\prec_y$). This implies that if a vertical (horizontal) segment of an edge shares a point with the interior of a box, then this box correspond to one of its endvertices. Further, any oblique segment is fully contained inside the box of its endvertex, hence crossings can only happen between a vertical and a horizontal segment \mbox{which implies that the drawing is RAC.}

To complete the proof of \cref{thm:7colorable}, we discuss the time complexity and the required area.
We apply \cref{thm:2-factors} to $G \setminus M$ to obtain $F_1$, $F_2$, $F_3$ in $\mathcal{O}(n)$ time. Computing the labels clearly takes $\mathcal{O}(n)$ time. For each cycle of $F_1$, the ordering of its internal vertices in \cref{lem:no-3-0} can be done in time linear in the size of the cycle by computing for each vertex the number of forward and backward edges, and of chords. Computing the tags takes $\mathcal{O}(n)$ time. In each of the following two iterations, we perform a constant number of operations per vertex. Hence we can conclude that the drawing can be computed in $\mathcal{O}(n)$ time.
For the area, we can observe that the size of the grid defined by the boxes is $8n \times 8n$ and by construction, any vertex and any bend point is placed on a point on the grid.

\begin{corollary}
Given a $7$-edge-colorable degree-$7$ graph with $n$ vertices and a \mbox{$7$-edge-coloring} of it, it is possible to compute in $O(n)$ time a $2$-bend RAC drawing of it with $O(n^2)$~area.
\end{corollary}

\section{Conclusions and Open Problems}
We significantly extended the previous work on RAC drawings for low-degree graphs in all reasonable settings derived by restricting the number of bends per edge to $0$, $1$, and $2$. The following open problems are naturally raised by our work.
\begin{itemize}
\item Are all $4$-edge-colorable degree-$3$ graphs RAC (refer to \cref{q:degree-3-rac})?
\item Are all degree-$5$ graphs $1$-bend RAC? \cref{fig:1-bend-degree-5-graphs} shows $1$-bend RAC drawings of two prominent degree-$5$ graphs, namely $K_{5,5}$ and the $5$-cube graph. What about degree-$6$ graphs?
\item Is it possible to extend \cref{thm:7colorable} to all (i.e., not 7-edge-colorable) degree-$7$ graphs or even to (subclasses of) graphs of higher degree, e.g. Hamiltonian degree-$8$ graphs?
\item While recognizing graphs that admit a (straight-line) RAC drawing is NP-hard~\cite{DBLP:journals/jgaa/ArgyriouBS12}, the complexity of the recognition problem in the $1$- and $2$-bend setting is still unknown.
\end{itemize}

\begin{toappendix}
\begin{figure}[h!]
    \centering
    \begin{subfigure}[b]{.48\textwidth}
    \centering
    \includegraphics[scale=0.8,page=1]{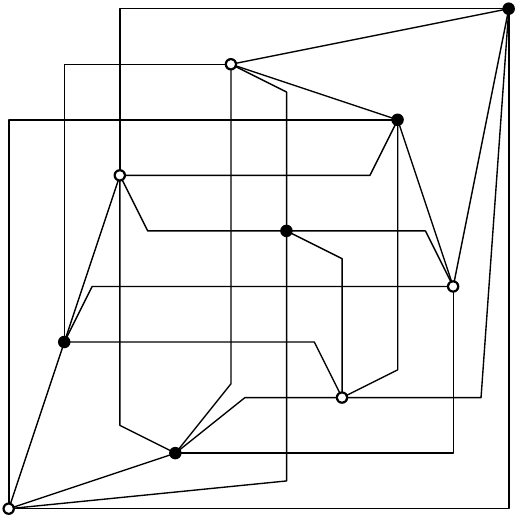}
    \caption{}
    \label{fig:k5,5}
    \end{subfigure}
    \hfil
    \begin{subfigure}[b]{.48\textwidth}
    \centering
    \includegraphics[scale=0.8,page=1]{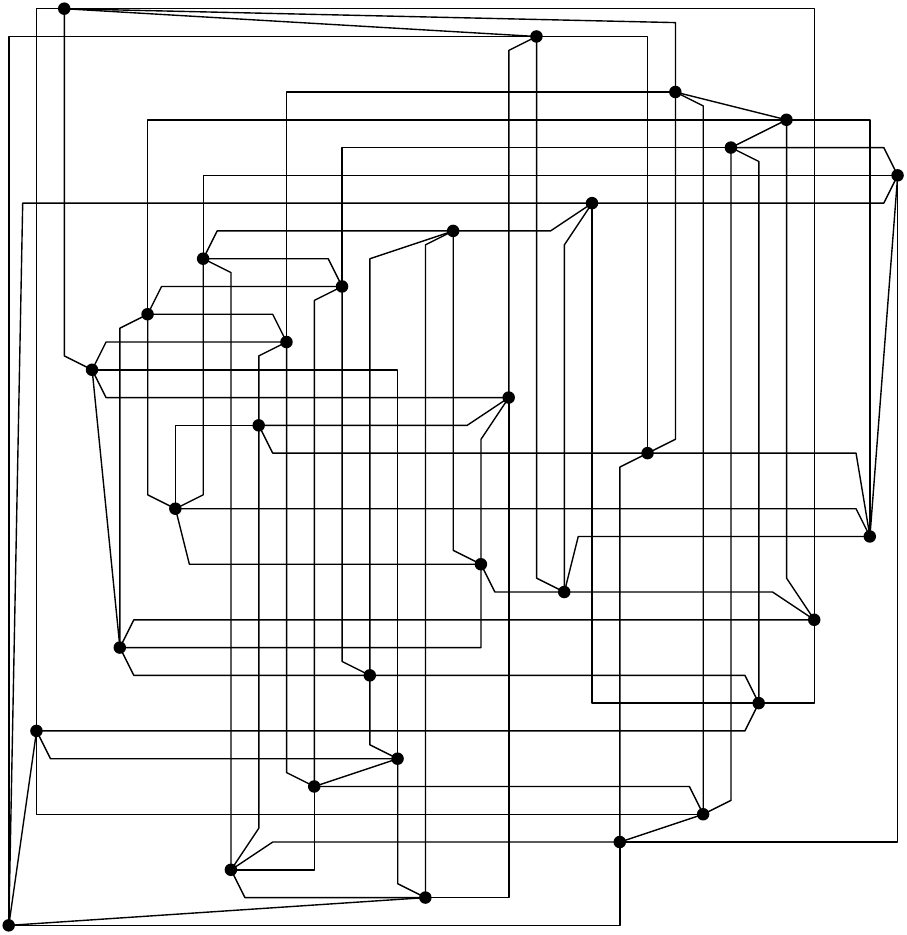}
    \caption{}
    \label{fig:5-cube}
    \end{subfigure}
    \caption{$1$-bend RAC drawings for (a)~the $K_{5,5}$ graph, and (b)~the $5$-cube graph.}
    \label{fig:1-bend-degree-5-graphs}
\end{figure}
\end{toappendix}

\bibliographystyle{plain}
\bibliography{bibliography.bib}

\end{document}